\documentclass{llncs}

\usepackage{graphicx}
\usepackage{amsmath,amssymb}
\usepackage{graphicx,subfigure,tabularx,multirow}
\usepackage{times,verbatim}
\usepackage{enumerate,url,color}
\usepackage[bottom]{footmisc}
\usepackage{hyperref}
\usepackage[labelfont=bf,size=small]{caption}
\newcommand{\blink}[1]{\textnormal{\texttt{#1}}}

\newcommand{\NPC}[0]{\blink{NP}-complete}
\newcommand{\NPH}[0]{\blink{NP}-hard}
\newcommand{\Frechet}[0]{Fr\'{e}chet}

\newcommand{\SCoP}{\textsc{SCP}}
\newcommand{\EDP}{\textsc{Disjoint Paths}}

\newcommand{\alg}[1]{\textsc{#1}}

\author{
L.~De La Cruz$^1$ \and
S.~Kobourov$^1$ \and
S.~Pupyrev$^{1,2}$ \and
P.~Shen$^1$ \and
S.~Veeramoni$^{1}$
}

\institute{
$^1$ Department of Computer Science, University of Arizona\\
$^2$ Institute of Mathematics and Computer Science, Ural Federal University
}

\title{Computing Consensus Curves}

\begin{document}

\maketitle

\begin{abstract}
We study the problem of extracting accurate average ant
trajectories from many (inaccurate) input trajectories
contributed by citizen scientists. Although there are many generic
software tools for motion tracking and specific ones for insect
tracking, even untrained humans are better at this task.
We consider several local (one ant at a time) and global (all ants
together) methods. Our best performing algorithm uses a novel global
method, based on finding edge-disjoint paths in a graph constructed from
the input trajectories. The underlying optimization problem is a new
and interesting network flow variant.
Even though the problem is NP-complete, two heuristics work well in practice,
outperforming all other approaches,
including the best automated system.
\end{abstract}

\section{Introduction}
\label{sec:intro}
Tracking moving objects in video is a difficult task to automate. Despite advances in
machine learning and computer vision, the best way to accomplish such
a task still is by hand. At the same time, people spend millions of hours each day
playing games like \emph{Solitaire}, \emph{Angry Birds}, and \emph{Farmville} on phones and
computers. This presents
an opportunity to harness some of the time people spend on online
games for more productive but still
enjoyable work.
In the last few years it has been shown that citizen
scientists can contribute to image processing tasks, such as the
\emph{Galaxy Zoo} project in which thousands of citizen
scientists labeled millions of images of galaxies from
the Hubble Deep Sky Survey and \emph{FoldIt} in
which online gamers helped to decode the structure of an AIDS protein~---
a problem which stumped researchers for 15 years.

\emph{AngryAnts} is our online game available at \url{http://angryants.cs.arizona.edu},
which plays videos and allows citizen scientists to build the
trajectory of a specified ant via mouse clicks;
see Fig.~\ref{fig:antcolony}.
When we have enough data, we compute a
most realistic {\em consensus} trajectory for each ant.
Our motivation comes from biologists who wish to discover
longitudinal behavioral patterns in ant colonies.
The trajectories of individual ants in a colony extracted from videos
are needed to answer
questions such as how often do ants communicate, what different roles
do ants play in a colony, and how does interaction and communication affect
the success or failure of a colony? Existing automated solutions are
not good enough, and there is
only so much data that even motivated students can annotate in the
research lab.

\begin{figure}[tb]
\centering
\includegraphics[height=6cm]{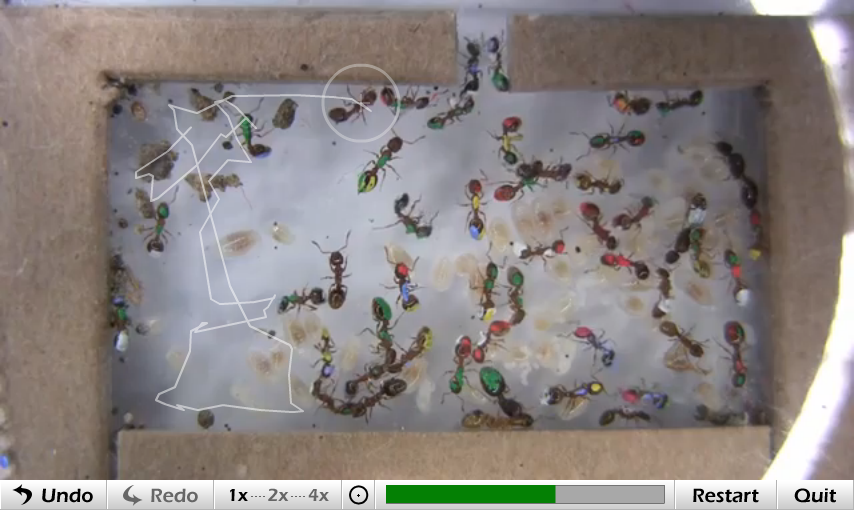}
\caption{\small A snapshot of the game environment with the selected ant in the circle.}
\label{fig:antcolony}
\end{figure}

{\bf Related Work:}
The problem of computing the most likely trajectory from a set of given trajectories has been studied in many different contexts and here we mention a few examples, which are similar to our approach.
Buchin {\em et~al.}~\cite{bbklsww-mt-10} look for
a representative trajectory for a given set of trajectories and
compute a median representative rather than the mean.
The \Frechet{} distance as similarity measure for
trajectories is studied by Buchin {\em et~al.}~\cite{DBLP:journals/gis/BuchinBG10}, who show how to
incorporate time-correspondence and directional constraints.
Trajcevski {\em et~al.}~\cite{DBLP:conf/gis/TrajcevskiDSTV07} use the
maximum distance at corresponding times as a measure of similarity between
pairs of trajectories, and describe algorithms for 
matching under rotations and translations.

Yilmaz {\em et~al.}~\cite{Yilmaz:2006:OTS:1177352.1177355} survey
the state-of-the-art in object tracking methods. Some of the most
recent methods include general approaches for tracking
cells undergoing collisions by Nguyen {\em et~al.}~\cite{5779709} and
specific approaches for tracking insects by
Fletcher {\em et~al.}~\cite{5711555}.
Also related are the automatic tracking
method for tracking bees by Veeraraghavan~{\em et~al.}~\cite{10.1109/TPAMI.2007.70707}
and cluster-based, data-association
approaches for tracking bats in infrared video by
Betke~{\em et~al.}~\cite{4270019}. Tracking the motion and interaction
of ants has also been studied by
Khan~{\em et~al.}~\cite{1512059}, who describe probabilistic
methods, and by Maitra~{\em et~al.}~\cite{5403051} using
computer vision techniques.

{\bf Our Contributions:}
We describe a citizen science approach for extracting
consensus trajectories in an online game setting; see
Fig.~\ref{fig:flow} for an overview.
Combining human-generated trajectories with a new global approach for
computing consensus curves outperforms
even the most sophisticated and computationally expensive tracking algorithms~\cite{Poff12}, even in the
most advantageous setting for automated solutions (e.g., high
resolution video, sparse ant colony, individually painted ants).

Consider a {\em trajectory} as a sequence of
$T$ pairs $(p_i,t_i), 1 \le i \le T$, where $p_i=(x_i, y_i)$ is a
point in the plane representing the position of an ant at timestamp $t_i$.
We assume that between timestamps an ant keeps a constant
velocity, and therefore, its trajectory is a polyline in 3D (or a possibly
self-intersecting polyline in 2D).
The input of our problem is a collection $\tau_1, \dots, \tau_m$ of
trajectories. Each trajectory corresponds to one of $k$ ants, and we
assume that there exists at least one trajectory for each ant, that
is, $m\ge k$. Since the input data comes from the \emph{AngryAnts} game, all
trajectories have the same length (number of points).
We also assume that the initial position of each ant is provided by
the game (from a different game level called ``Count the Ants'' where
users click on all the ants in the first frame of the video to identify their starting
positions). Therefore, the first points of trajectories corresponding
to the same ant are identical. Our goal is to compute a {\em consensus
trajectory}, that is, our best guess for the actual route taken, for
each of the $k$ ants. While intuitively we are looking for the most probable ant trajectories,
it is far from obvious how to measure the quality of a solution.

\begin{figure}[t]
\centering
\includegraphics[height=5.5cm]{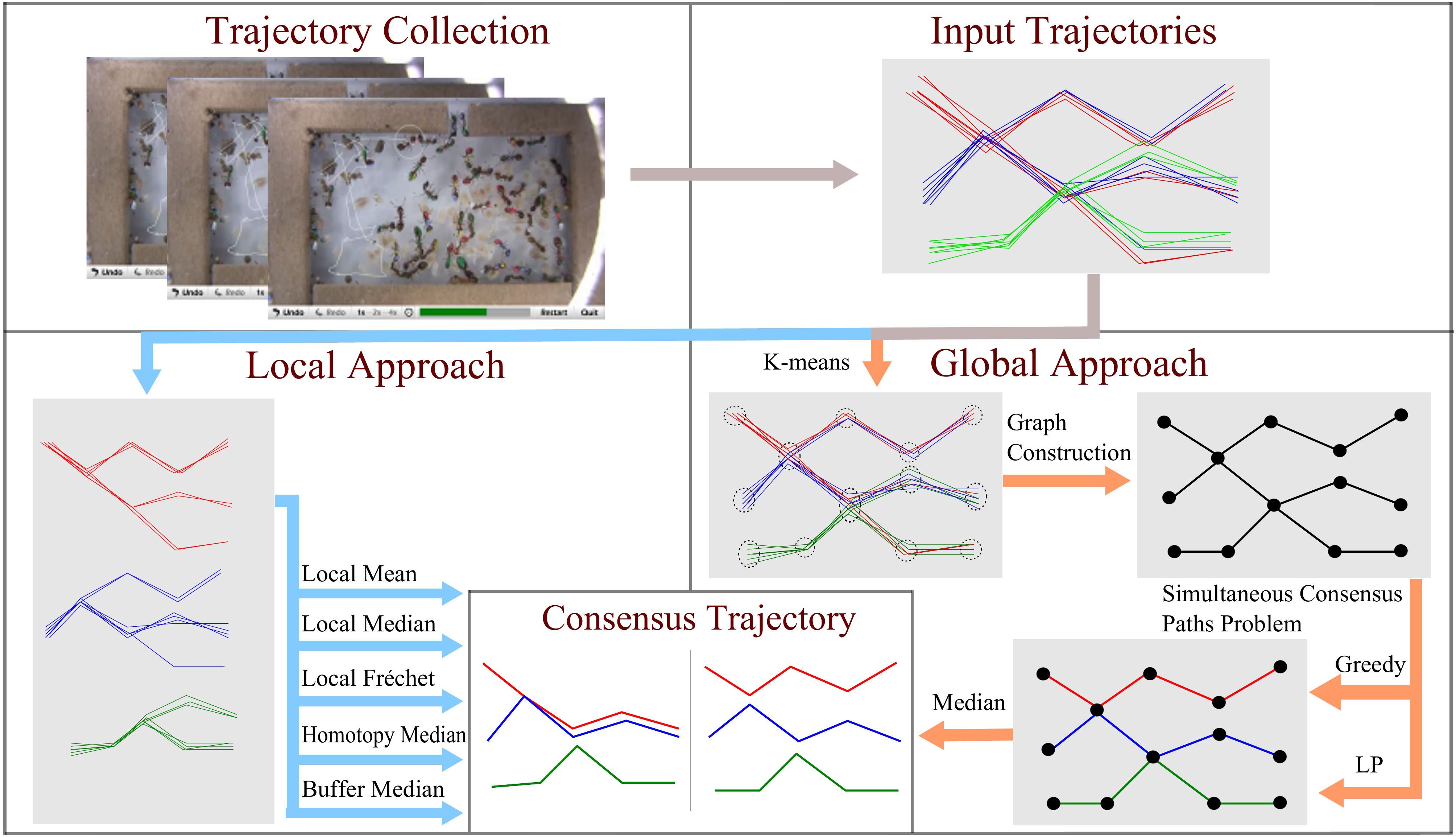}
\caption{Overview of consensus trajectory computation via citizen
science. Local consensus may route different ants along  the
same trajectory (red and blue), while the global consensus ensures
disjoint trajectories.}
\label{fig:flow}
\end{figure}
We designed, implemented, and evaluated several methods for
computing accurate consensus trajectories from many (possibly
inaccurate) trajectories submitted by citizen scientists. We
consider several local strategies (clustering, median
trajectories, and \Frechet{} matching) in one-ant-at-a-time setting. We
also designed and implemented
a novel global
method in the all-ants-together setting. This approach is
based on finding edge-disjoint paths in a graph
constructed from all input trajectories. The underlying optimization
problem is a new and interesting variant of network flow. Even though
the problem is \NPC{}, our two heuristics work
well in practice, outperforming all other approaches including the
automated system.

\section{The Local Approach}
\label{sect:local}

\noindent{\bf Local Mean:}
Intuitively, a mean trajectory is the one that averages locations for
input trajectories. To compute the mean, we identify all input
trajectories $\tau_1, \dots, \tau_{m_c}$ for a particular ant $c$. For
each timestamp $t_i$, we query points
$p_1=(x_1, y_1), \dots, p_{m_c}=(x_{m_c}, y_{m_c})$ corresponding to
$t_i$. The average point is $(x_{avg}, y_{avg})$,
where $x_{avg}=(x_1 + \dots + x_{m_c})/m_c$ and $y_{avg}=(y_1 + \dots
+ y_{m_c})/m_c$ gives the location of the mean
trajectory at timestamp $t_i$. The sequence of these average points
over time defines the local mean consensus trajectory, which
 is good when the number of input trajectories is
very large or when all input trajectories are very accurate. In
reality, however, this is often not the case; a single
inaccurate input trajectory may greatly influence the result.
\medskip

\noindent{\bf Local Median:}
The median point is more robust to outliers than the mean.
For a set of points $p_1,p_2, \dots, p_{m_c}$
the median is
the point $(x_{med}, y_{med})$, where $x_{med}$ is the median of the
array $x_1, \dots, x_{m_c}$ and $y_{med}$ is
the median of the array $y_1, \dots, y_{m_c}$. The sequence of these
average points over time defines the local median consensus trajectory.
Note that the median of a set of points is not necessarily a point of
the set: it could place an ant at a point that matches none of the input trajectories.

\medskip

\noindent{\bf Local \Frechet{}:}
Informally, the \Frechet{} distance between two trajectories can be illustrated as the
minimum dog-leash distance that allows a
man to walk along one trajectory and his dog along the other, while connected at all times by the leash~\cite{DBLP:journals/ijcga/AltG95}. Computing
the \Frechet{} distance produces an alignment of the trajectories:
at each step, the position of the man is mapped to the position of the
dog. Given the \Frechet{} alignment of two trajectories, we compute their
consensus by taking the midpoint of the leash over time.
To find the consensus for all input trajectories, we repeatedly
compute pairwise consensuses until only one trajectory remains.
Since the results of the algorithm depends on the order in which the
trajectories are merged, we try several different random orders
choosing the best result. Note that, by definition, the
\Frechet{} alignment of trajectories uses only the order of points and
ignores the timestamps that are an essential feature of our input.

\begin{figure}[t]
\center
	\subfigure[]{
    \includegraphics[height=1.6cm]{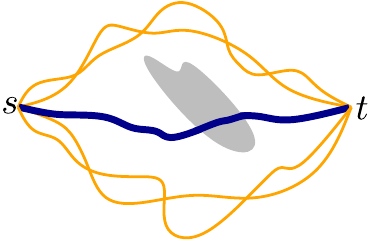}
    \label{fig:median1}}
	\subfigure[]{
    \includegraphics[height=1.6cm]{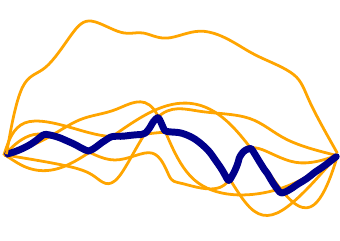}
    \label{fig:median2}}
	\subfigure[]{
    \includegraphics[height=1.6cm]{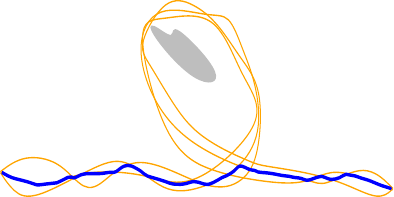}
    \label{fig:median3}}
	\subfigure[]{
    \includegraphics[height=1.6cm]{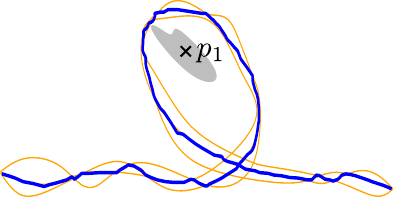}
    \label{fig:median4}}
\caption{\small (a)~~Simple average of trajectories may result in a consensus
 that goes through an obstacle.
(b)~The median trajectory always follows some piece of input
trajectory and is robust to outliers.
(c)~The simple median algorithm might make a shortcut and miss a self-
intersecting loop.
(d)~Homotopy median with an added obstacle to avoid short cutting.}
\end{figure}

\medskip\noindent{\bf Homotopy Median and Buffer Median:}
In the above approaches, the average of two locations could be an
invalid location; see Fig.~\ref{fig:median1}.
In the median trajectory approaches, the computed trajectory always
lies on segments of input curves; see Fig.~\ref{fig:median2}. Note that
this is more restrictive than the point-based median method described above.
We use two algorithms suggested by Buchin~{\em et al.}~\cite{bbklsww-mt-10} and
Wiratma~\cite{wiratma10}.
In the \emph{homotopy median} algorithm, additional obstacles are placed in
large faces bounded by segments
of input trajectories to ensure that the median
trajectory is homotopic to the set of trajectories that go around
the obstacles. Hence, the median does not miss segments if the input
trajectories are self-intersecting; see Fig.~\ref{fig:median3} and
Fig.~\ref{fig:median4}.
The \emph{buffer median} algorithm is a combination of the buffer
concept and Dijkstra's shortest path algorithm. A buffer is defined around a
segment so that if the segment is a part of the
median trajectory, then its buffer intersects all input trajectories. Thus,
segments located near the middle of the trajectory have smaller
buffer size and are good candidates for the median
trajectory.
Note again that the homotopy median and buffer median approaches ignore the
timestamps that are an essential feature of our input.

\section{The Global Approach}
\label{sect:global}

In the global approach we consider all input trajectories for all $k$
ants together. The main motivation is that a trajectory corresponding
to an ant may contain valid pieces of trajectories for other ants:
a citizen scientist may mistakenly switch from tracking
an ant $x$ to tracking a different ant $y$ at an intersection point where
$x$ and $y$ cross paths. However, even when such mistakes occur, the
trace after the intersection point is still useful as it contains a
part of the trajectory of ant $y$. The global approach allows us to
retain this possibly useful data as shown by an example in
Fig.~\ref{fig:flow}.
Given a set of input trajectories, we compute the consensus
trajectories in three steps: (1)~create a graph $G$; (2)~compute edge-disjoint
paths in $G$; (3)~extract consensus trajectories from the paths.

\noindent{\bf Step 1}: We begin by creating a graph that
models the interactions between ants in the video. For every timestamp,
the graph has at most $k$ vertices, which correspond to the
positions of the $k$ ants. If several ants are located close to each
other, then
we consider them to be at the same vertex. Intuitively, each such vertex is
a possible point for a citizen scientist to switch to a wrong ant.

A weighted directed graph $G=(V,E)$ is constructed as follows. For
every timestamp $t_i$, we extract points $p_1,\dots,p_m$ from the
given trajectories, where $p_j$ is the position of trajectory $\tau_j$
at $t_i$. Using a modification of the $k$-means clustering algorithm~\cite{lloyd82}, we
partition the points into $\le k$ clusters.
Our clustering algorithm
differs from the classical $k$-means in that it always merges the points
into one cluster located closer than $50$ pixels from each other.
The vertices of $G$ are the clusters for all
timestamps; thus, $G$ has at most $kT$ vertices. We then add edges
between vertices in consecutive timestamps. Let $V_i$ represent a set
of vertices at timestamp $t_i$. We add an edge $(u,v)$ between two
vertices $u\in V_i$ and $v\in V_{i+1}$ if there is an input trajectory
with point $p_i$ belonging to cluster $u$ and point $p_{i+1}$
belonging to cluster $v$. For each edge, we create $k$
non-negative weights. For each ant $1\le x \le k$ and for each edge
$(u,v) \in E, u \in V_i, v \in V_{i+1}$, there is weight $w_{uv}^x\in
\mathbb{Z}_{\ge0}$, which equals to the number of trajectories between timestamps
$i$ and $i+1$ associated
with the ant $x$ passing through the clusters $u$ and $v$.

Note that by construction graph $G$ is acyclic. Let $d^{in}(v)$ and
$d^{out}(v)$ denote the indegree and the outdegree of the vertex $v$.
Clearly, the only vertices with $d^{in}(v)=0$ are in $V_1$, and the
only vertices with $d^{out}(v)=0$ are in $V_T$; we call them
\emph{source} and \emph{destination} vertices and denote them by $s_i$ and $t_i$, respectively. We assume that
for all \emph{intermediate} vertices $v\in V_i, 1<i<T$, we have
$d^{in}(v)=d^{out}(v)$; this is a realistic assumption
because for each intersection point of trajectories,
the number of incoming ants equals the number of outgoing ants.
We say that a directed graph satisfies the {\em ant-conservation
condition} if (1) the number of outgoing edges from sources and the
number of incoming edges to destinations is equal to $k$, that is,
$\sum_i d^{out}(s_i) = \sum_i d^{in}(t_i) = k$, and
(2) indegree and outdegree of all its intermediate vertices
are the same, that is, $d^{in}(v)=d^{out}(v)$.

\noindent{\bf Step 2:} We compute $k$ edge-disjoint paths in $G$,
corresponding to the most ``realistic'' trajectories of the ants. The
paths connect the $k$ distinct vertices in $V_1$ to the $k$ distinct
vertices in $V_T$. Since the initial position of each ant is a part of
input, we know the starting vertex of each path. However, it is not
obvious which of the destination vertices in $V_T$ correspond to each
ant. To measure the quality of the resulting ant
trajectories, that is, how well the edge-disjoint paths match the
input trajectories, we introduce the following optimization problem.

\noindent\textbf{\textsc{Simultaneous Consensus Paths} (SCP)}\\
\indent\emph{Input}: A directed acyclic graph $G=(V,E)$ with $k$
sources $s_1,\dots,s_k$ and $k$ destinations $t_1,\dots,t_k$
satisfying the ant-conservation condition.
The weight of an edge $e\in E$ for $1\le i \le k$ equals $w_e^i \in
\mathbb{Z}_{\ge 0}$.\\
\indent\emph{Task}: Find $k$ edge-disjoint paths $P_1,\dots,P_k$ so
that path $P_i$ starts at $s_i$ and ends at $t_j$ for some $1\le j \le k$,
and the total
$cost=\sum_{i=1}^k \sum_{e\in P_i} w_e^i$ is maximized. Note that the
objective is to simultaneously optimize all $k$ disjoint paths.
The decision problem is to find $k$ edge-disjoint paths with total $cost \ge c$ for some constant $c \ge 0$.

The problem is related to the integer multi-commodity flow problem,
which is known to be \NPH{}~\cite{Even75}. In our setting, the weights
on edges for different ``commodities'' are different, and the
source-destination pairs are not known in advance.
We study the \SCoP{} problem in the next section.

\noindent{\bf Step 3:}
We construct consensus trajectories corresponding to a solution of
the \SCoP{} problem. Let $P$ be a path for an ant $x$ computed
in the previous step. For each timestamp $t_i$, we consider the edge
$(u,v)\in P, u\in V_i, v\in V_{i+1}$, and find a set $S_{uv}$ of all
input trajectories passing through both clusters $u$ and $v$. We
emphasize here that $S_{uv}$ may contain (and often does contain)
trajectories corresponding to more ants than just ant $x$. Next we
compute the median of points of $S_{uv}$ at timestamp $t_i$; the
median is used as the position of ant $x$ at timestamp $t_i$. The
resulting trajectory of $x$ is a polyline connecting its positions for
all $t_i, 1\le i \le T$.

\section{The \textsc{Simultaneous Consensus Paths} Problem}
\label{sect:antprob}

The \SCoP{} problem is \NPC{} even when restricted to planar graphs by a reduction from a
variant of the edge-disjoint path problem~\cite{Marx04}.
On the other hand, the problem is fixed-parameter tractable in the number of paths.
For a fixed constant $k$, it can be solved optimally via dynamic programming
in time $O(|E| + k!|V|)$.

\subsection{Hardness Result}

Let a \emph{grid graph} be one that is a subgraph of the
rectangular grid, that is, a graph with $n\times m$ vertices such that
$v_{i,j}$ is connected to $v_{i',j'}$ if and only if $|i-i'|=1$ and
$j=j'$, or $i=i'$ and $|j-j'|=1$. It is easy to see that any grid
graph is planar. We prove that the \SCoP{} problem is \NPC{} even when
restricted to grid graphs by
a reduction from a variant of the edge-disjoint path problem, which is
formulated as follows.

\medskip

\noindent\textbf{\EDP{}}\\
\indent\emph{Input}: A directed acyclic grid graph $G$ satisfying the
ant-conservation condition and a set of $k$ source-destination pairs
$(s_1,t_1),\dots,(s_k,t_k)$. \\ \indent\emph{Task}: Find $k$
edge-disjoint paths $P_1,\dots,P_k$ so that $P_i$ starts at $s_i$ and
ends at $t_i$.

There are two differences with the \SCoP{} problem: (1) in the
\EDP{} problem, a source-destination pair is fixed for every path; (2)
the goal is to construct $k$ paths, while the \SCoP{}
problem is an optimization of a weighted sum.
Deciding whether the \EDP{} problem has a solution
with exactly $k$ paths is \NPC{}~\cite{Marx04}.
We also note that the
graphs used in the construction of the proof satisfy the ant-conservation
condition.

\begin{theorem}
\label{thm:npc}
The \SCoP{} problem is \NPC{} on acyclic directed grid graphs.
\end{theorem}

\begin{proof}
It is easy to verify a solution of \SCoP{} in polynomial time;
we argue the hardness result via standard reduction.
Let $G=(V,E)$, $(s_1,t_1),
\dots, (s_k,t_k)$ be an instance of the \EDP{} problem. We create a
new directed acyclic graph $G'$ satisfying the ant-conservation
condition. Then we assign weights for the edges of $G'$ so that an
optimal solution for the \SCoP{} problem on $G'$ connects as many
pairs $(s_i,t_i)$ as possible.
$G'$ is constructed from $G$ by adding extra
$k$ sources and $k$ paths of length $k|V|$, which we refer to
as tails; see Fig.~\ref{fig:gprime}. In the new graph $V'=V \cup
\{s_1',\dots,s_k'\} \cup T_1 \cup \dots \cup T_k$, where every tail
$T_i = \{t_{i_1}, t_{i_2}, \dots, t_{i_{k|V|}}=t_i'\}$. Edges in $G'$
connect the new sources and the tails with the vertices of $G$:
$$E'=E \cup \bigcup_{i=1}^k (s_i',s_i) \cup \bigcup_{i=1}^k \{(t_i
,t_{i_1}),(t_{i_1} ,t_{i_2}), \dots,(t_{i_{k|V|-1}},t_{i_{k|V|}})\}.$$
Every maximal path (in terms of edges) in $G'$ has its dedicated source, and no two
edge-disjoint paths end at the same destination. It is easy to see
that $G'$ is acyclic and satisfies the ant-conservation condition.

In the \SCoP{} problem, the paths do not have fixed destinations. To
make sure that a path starts at $s_i'$ and ends at $t_i'$, assign
heavy weights to the edges on the tail of the $i$-th path:
For every $1\le i \le k$, choose some
path $R_i$ from $s_i'$ to $t_i'$. For all edges $e\in R_i$, set
$w^i_e = 1$, and for the remaining edges, $e\not\in R_i$, set $w^i_e
= 0$.
 Note that these paths $R_i$ need not be disjoint; we use them only to
assign the edge weight; see Fig.~\ref{fig:gprime2}.

\begin{figure}[h!]
    \center
	\subfigure[]{
    \includegraphics[height=2.6cm]{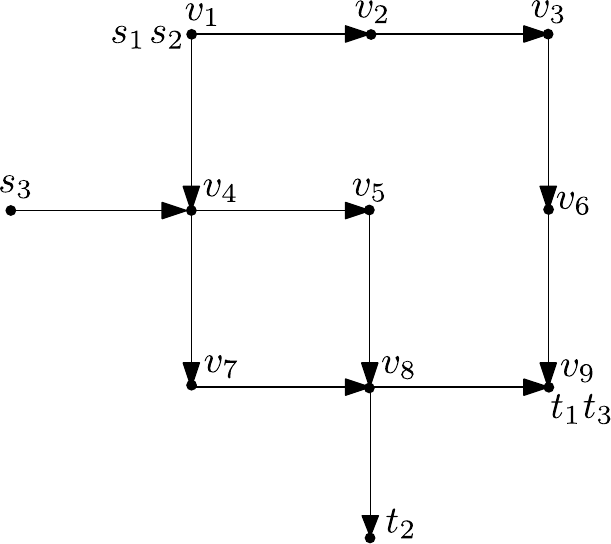}}
	\subfigure[]{
    \includegraphics[height=2cm]{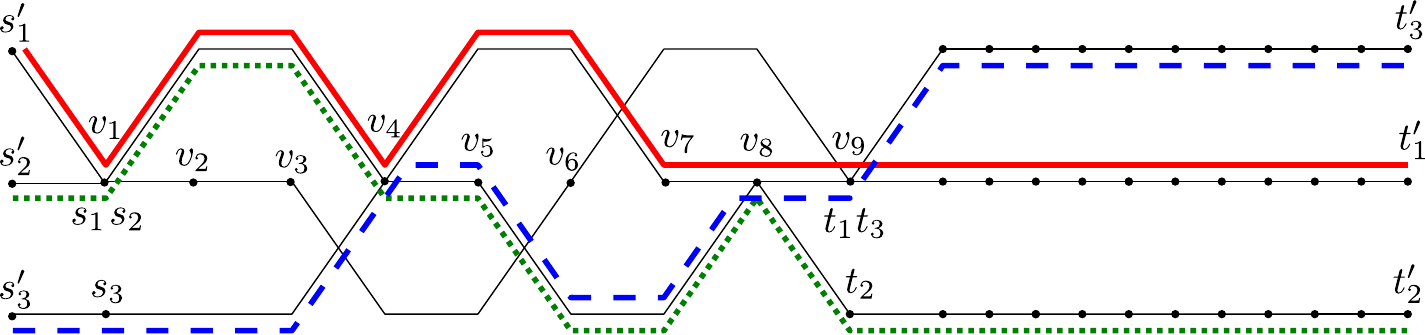}
    \label{fig:gprime2}}
    \caption{\small (a)~The input grid graph $G$ for the \EDP{} problem with
    3 source-destination pairs.
    (b)~Construction of the graph $G'$. Paths $R_1$, $R_2$, and $R_3$
    are colored red (solid), green (dotted), and blue (dashed), respectively.
    The tail connecting $t_i$ with $t_i'$ for $1\le i\le 3$ forces a
    path $P_i$ to end at $t_i$ in the optimal
    solution of the \SCoP{} problem.}
    \label{fig:gprime}
\end{figure}

Now we prove that there are $k$ edge-disjoint paths in $G$ if and only
if there is a solution of the \SCoP{} problem for $G'$ with $cost \ge
k^2|V|$. If there exist $k$ edge-disjoint paths in $G$ connecting
$s_i$ to $t_i$, then we can use these paths to construct heavy paths in
$G'$. Let $P_i=s_i\dots t_i$ be a path in $G$. Path $P_i'$ is
constructed from $P_i$ by adding a source and a tail:
$P_i'=s_i's_i\dots t_i\dots t_i'$. Since $w^i_e=1$ for the edges $e$
than $k|V|$. Paths $P_i'$ form a solution for the problem with
$cost\ge k^2|V|$.

Conversely, suppose the edge-disjoint paths $P_i'$ are a solution for
the \SCoP{} problem with $cost\ge k^2|V|$. We show that every $P_i'$
passes through $s_i$ and $t_i$ and, therefore, subpaths $P_i=P_i'\cap
V$ (a subgraph of $P_i'$ induced by $V$)
comprise a solution for the \EDP{} problem. For the sake of
contradiction, suppose $P_i'$ does not pass through $t_i$, that is, it
ends at $t_j'$ for some $j\neq i$. Then $P_i'$ passes through the
``wrong tail'': $P_i'=s_i's_i\dots t_j \dots t_j'$. The weight of
$P_i'$ is at most $|V|$ since there are at most $V$ edges in $E$ with
weight $w^i_{e}=1$. Symmetrically, path $P_j'$ has weight at most
$|V|$. Since the maximum weight on edges of $E$ is $k|V|$, the paths
$P_i'$ have weight at most $k|V| + k(k-2)|V| < k^2|V|$, thus
contradicting our assumption.

\end{proof}

\subsection{Exact Algorithm for Few Paths}

The next theorem implies that the \SCoP{} problem is
fixed-parameter tractable in the number of paths.
Note that if a graph $G$ satisfies the ant-conservation condition, then
in any solution with $k$ edge-disjoint paths all the edges of $G$
are covered by a path~\cite{Marx04}.

\begin{theorem}
\label{thm:fpt}
For a fixed constant $k$, the \SCoP{} problem can be solved optimally
in time $O(|E| + k!|V|)$.
\end{theorem}

\begin{proof}
Let $G=(V,E)$ be the input directed acyclic weighted graph,
$s_1,\dots,s_k \in V$ be the sources, and $t_1,\dots,t_k \in V$ be the
destinations.
We solve the problem using dynamic programming.
First, we compute a topological order of the vertices of $G$ and
fix the resulting order $u_1,\dots,u_{|V|}$. For convenience, we add
a super-destination vertex $t$ to $G$ together with zero-weight edges
$(t_1,t), \dots, (t_k,t)$.
Let $G^i=(V^i,E^i)$ for $1 \le i \le |V|+1$ be a subgraph of $G$ induced
by the vertices $u_i,\dots,u_{|V|},t$. For each $0\le i \le |V|+1$ we construct
a multiset $T_i$ with $|T_i|=k$. The multiset
$T_0$ consists of sources of $G$ so that $s_i$ is present $d^{out}(s_i)$ times.
For every $1\le i \le |V|$, we build $T_{i+1}$ from $T_i$ by removing vertex
$u_i$ and adding its outgoing neighbors. It easy to see that $T_{|V|+1}=(t,\dots,t)$.

Let us call the \emph{state} corresponding to $V^i$ an ordered sequence of
elements (not necessarily distinct) of $T_i$;
that is, $(v_1,\dots,v_k)$ in
which $v_j\in T_i, 1\le j\le k$ for some $i$.
For a state $(v_1,\dots,v_k)$, let $F(v_1,\dots,v_k)$ be the optimal
cost of routing $k$ edge-disjoint paths on $G$ so that the $i$-th path
starts at $v_i$ and ends at $t$.
It is easy to see that $F(t,\dots,t)=0$, and the optimal solution for the
\SCoP{} problem is $F(s_1,\dots,s_k)$.

In order to compute the value $F(S)$ for the state
$S=(v_1,\dots,v_k)$ corresponding to $V^i$, we choose the vertex $u_i\in S$.
Consider the edges $e_1=(u_i,a_1),\dots,e_d=(u_i,a_d)$, where $d=d^{out}(u_i)$.
Since all these edges should be covered by paths, the vertex $u_i$ is
repeated  exactly $d^{in}(u_i)=d$ times  in $S$;
without loss of generality we may assume that $S=(u_i,\dots,u_i,v_{d+1},\dots,v_k)$.
The paths may cover the edges in $d!$ different ways:
$$
F(v_1,\dots,v_d,v_{d+1},\dots,v_k) = \max_{\pi} \left( F(a_{\pi_1},
\dots, a_{\pi_d}, v_{d+1}, \dots, v_k) +
\sum_{j=1}^k w^{\pi_j}_{e_j} \right),
$$
where $\pi$ runs over all permutations of $\{1,\dots,d\}$. Note that
$(a_{\pi_1}, \dots, a_{\pi_d}, v_{d+1}, \dots, v_k)$ is a state
corresponding to $T_j$ with $j=i+1$.

The correctness of the algorithm follows from the observation that  we
consider all possible ways to route the paths at every intermediate
vertex.
The number of different states corresponding to $V^i$ is $k!/d!$, where $d$ is
the outdegree of the vertex $u_i$. In order to compute the values $F(S)$
for a state corresponding to $V^i$ we
check $d!$ permutations; hence, the running time is $k!$ for each
$V^i$. Summing over all graphs $G_i$, we obtain $O(|E|+k!|V|)$ running time.

\end{proof}

Next we suggest a greedy heuristic and provide an integer linear programming (ILP) formulation for \SCoP{}.

\begin{figure}[t]
    \center
	\subfigure[]{
    \includegraphics[height=3cm]{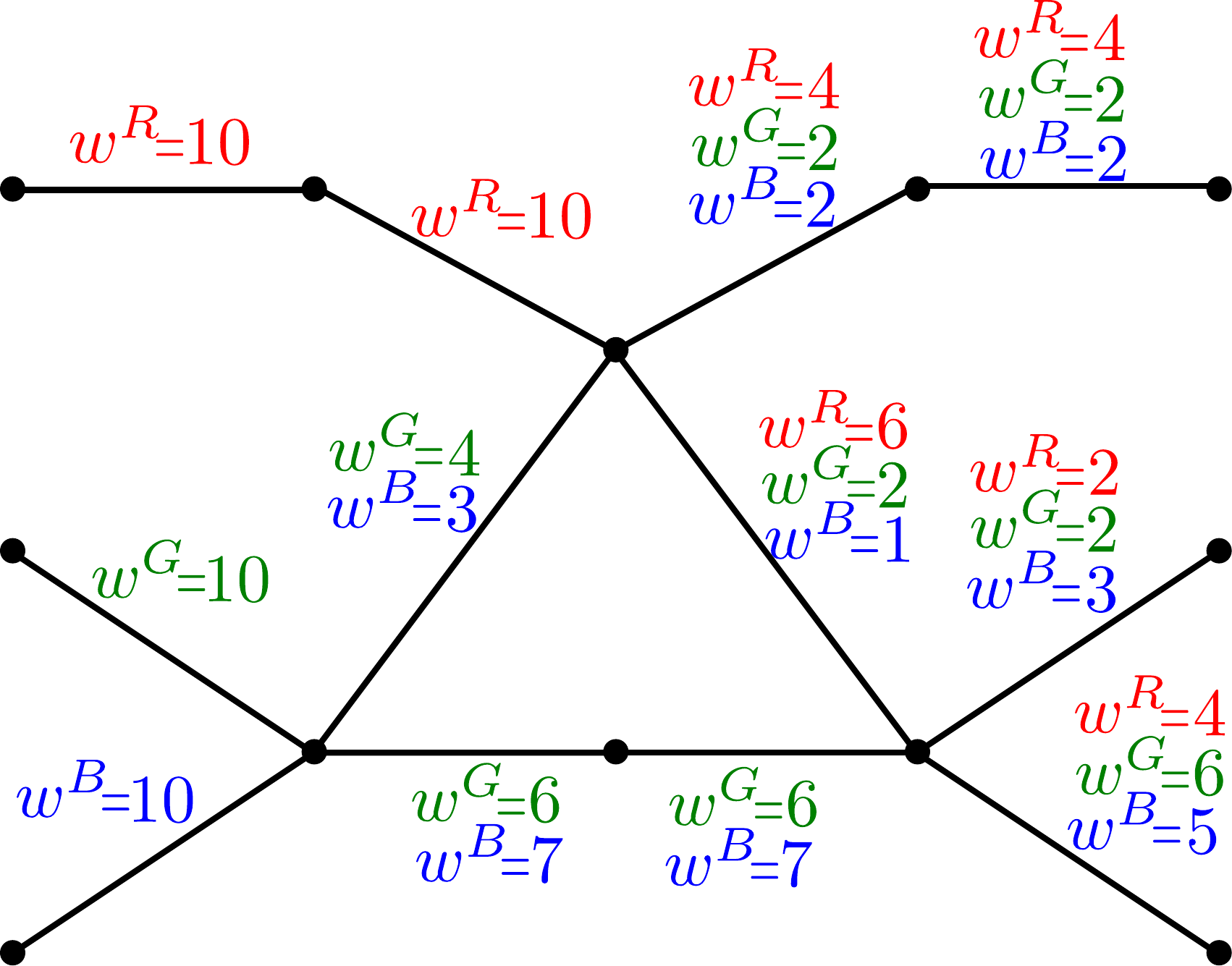}}
\hfill
	\subfigure[]{
    \includegraphics[height=3cm]{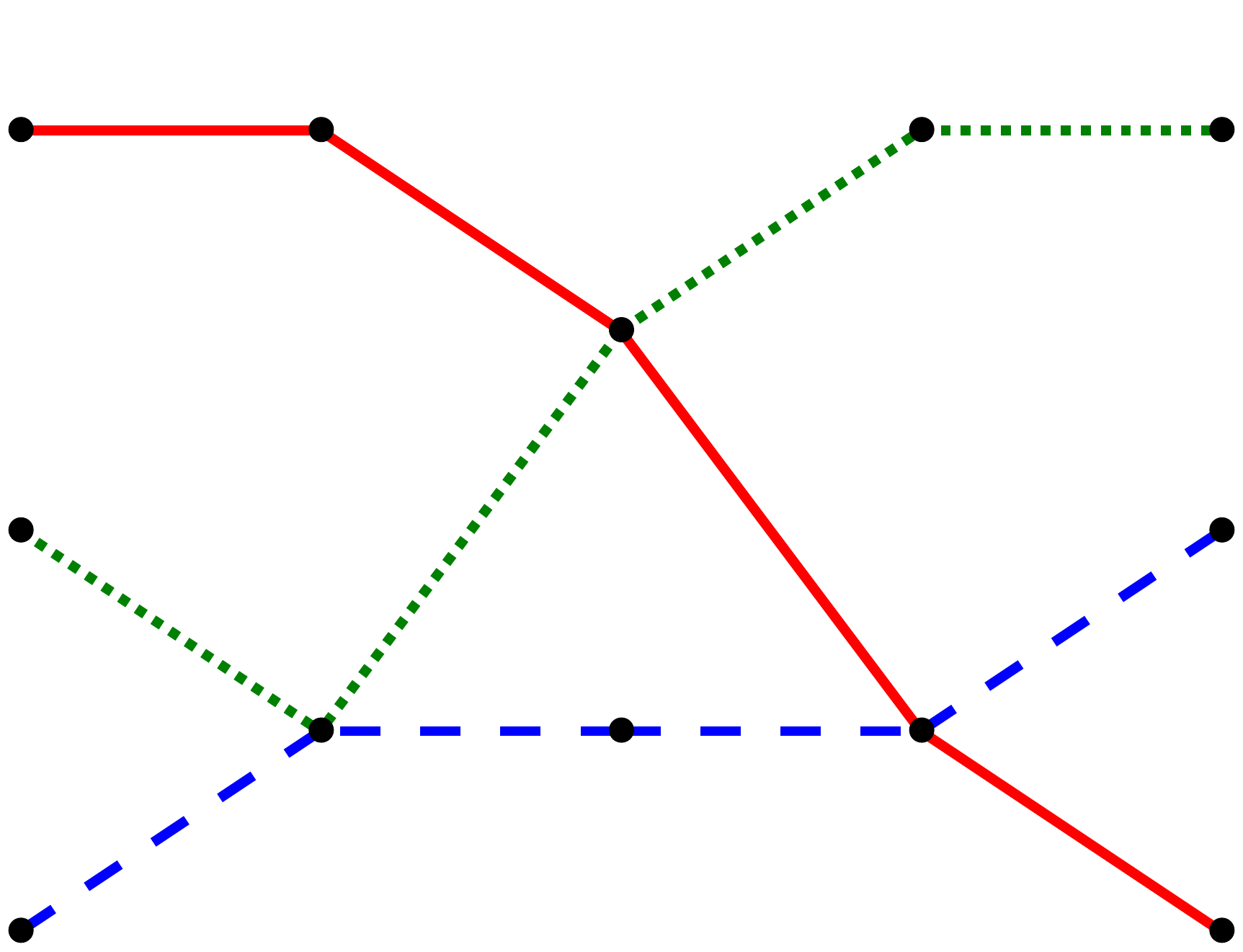}}
\hfill
	\subfigure[]{
    \includegraphics[height=3cm]{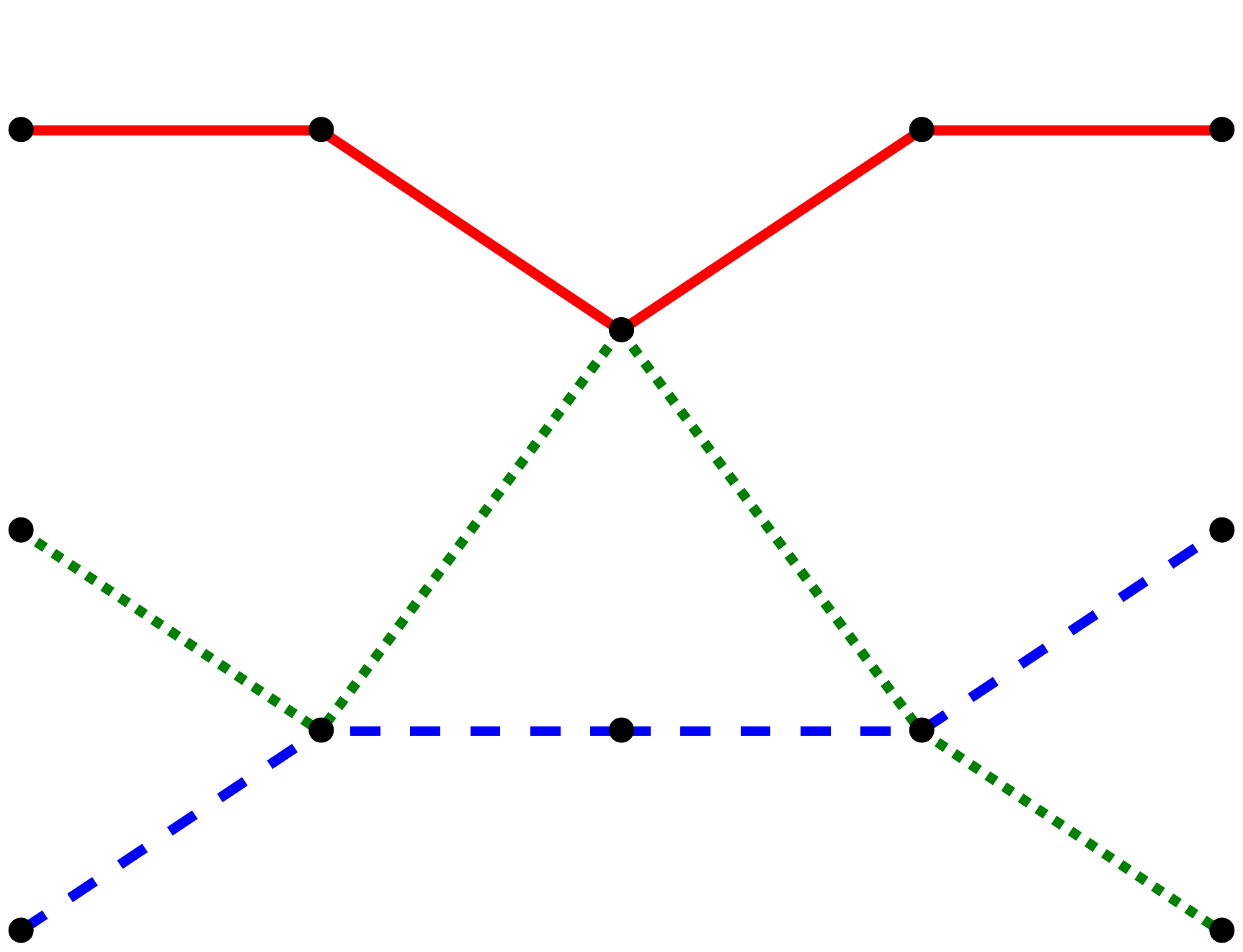}}
    \caption{The greedy algorithm may produce non-optimal solution.
    (a)~The input graph with 3 paths: R, G, and B.
    (b)~The paths computed by the greedy algorithm with $cost=75$
    (R is shown solid, G~-- dotted, B~-- dashed).
    (c)~The optimal solution with $cost=77$.}
    \label{fig:greedywrong}
\end{figure}

{\bf The Greedy Algorithm:}
The algorithm finds the longest path among all source-destination pairs $s_i$, $t_j$.
The length of a path starting
at $s_i$ is the sum of $w^i_{uv}$ for all edges $(u,v)$ of the path.
Since $G$ is an acyclic directed graph, the longest path for the
specified pair $s_i, t_j$ can be computed in time $O(|E|+|V|)$ via
dynamic programming. Once the longest path is found, we remove all
edges of the path from $G$ and proceed with the next longest path. The
algorithm finds at most $k$ paths on each iteration and the number of
iterations is $k$; hence, the overall running time is
$O(k^2(|V|+|E|))$.

The algorithm always yields a solution with $k$ disjoint paths.
Initially, $G$ satisfies the ant-conservation condition:
$d^{in}(v)=d^{out}(v)$ for all intermediate vertices $v$.
Since $G$ is connected, there exists a source-destination path. After removing the longest path, $G$
may be disconnected, but the ant-conservation condition still holds
for every connected component.
The number of outgoing edges from sources
and the number of incoming edges to destinations are equal for every
connected component.
Thus, the greedy algorithm produces a feasible solution but not necessarily the optimal one; see
Fig.~\ref{fig:greedywrong}.

\begin{figure}[t]
\center
        \subfigure[]{
    \includegraphics[height=2.8cm]{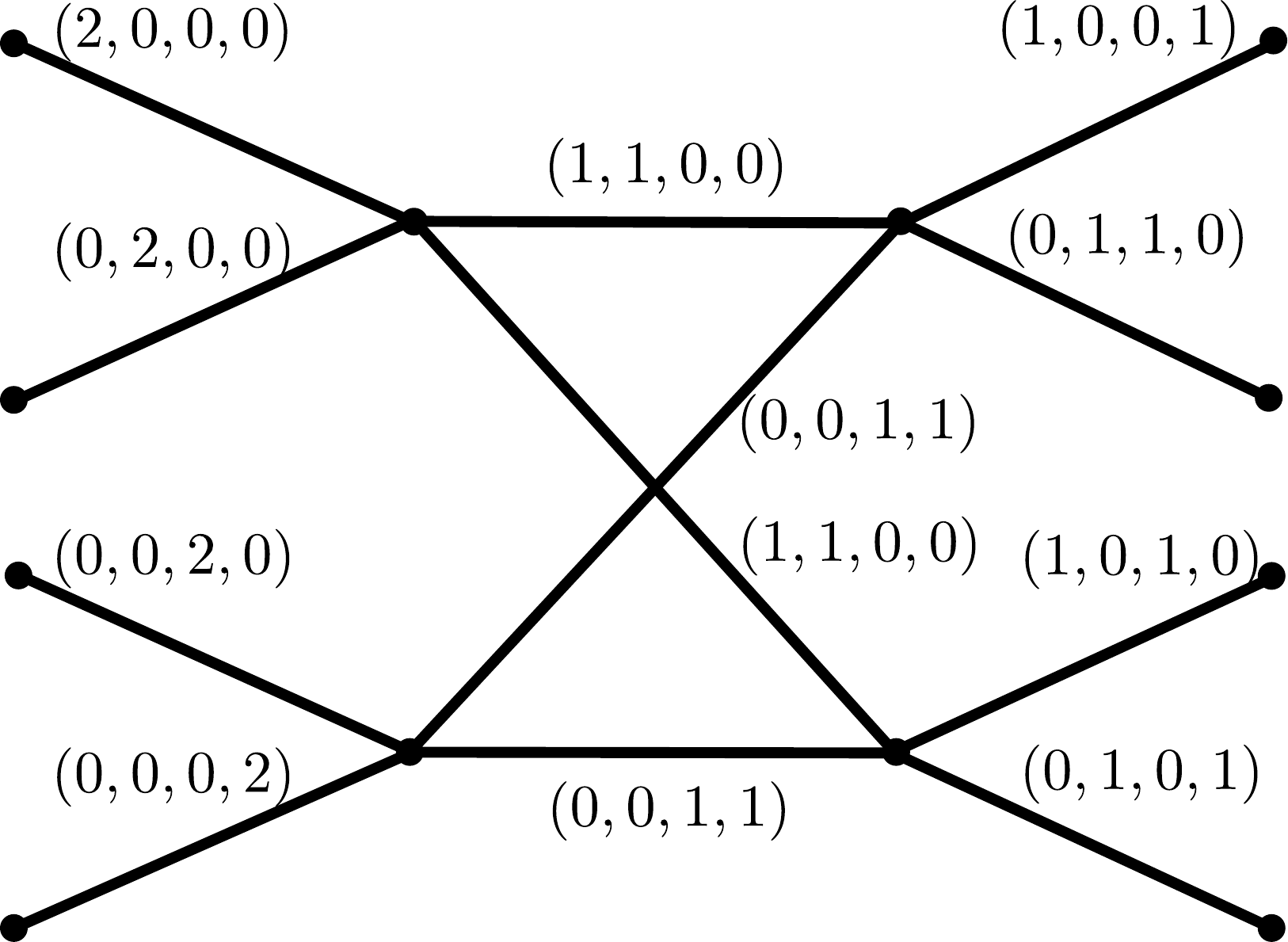}}
\hfill
        \subfigure[]{
    \includegraphics[height=2.8cm]{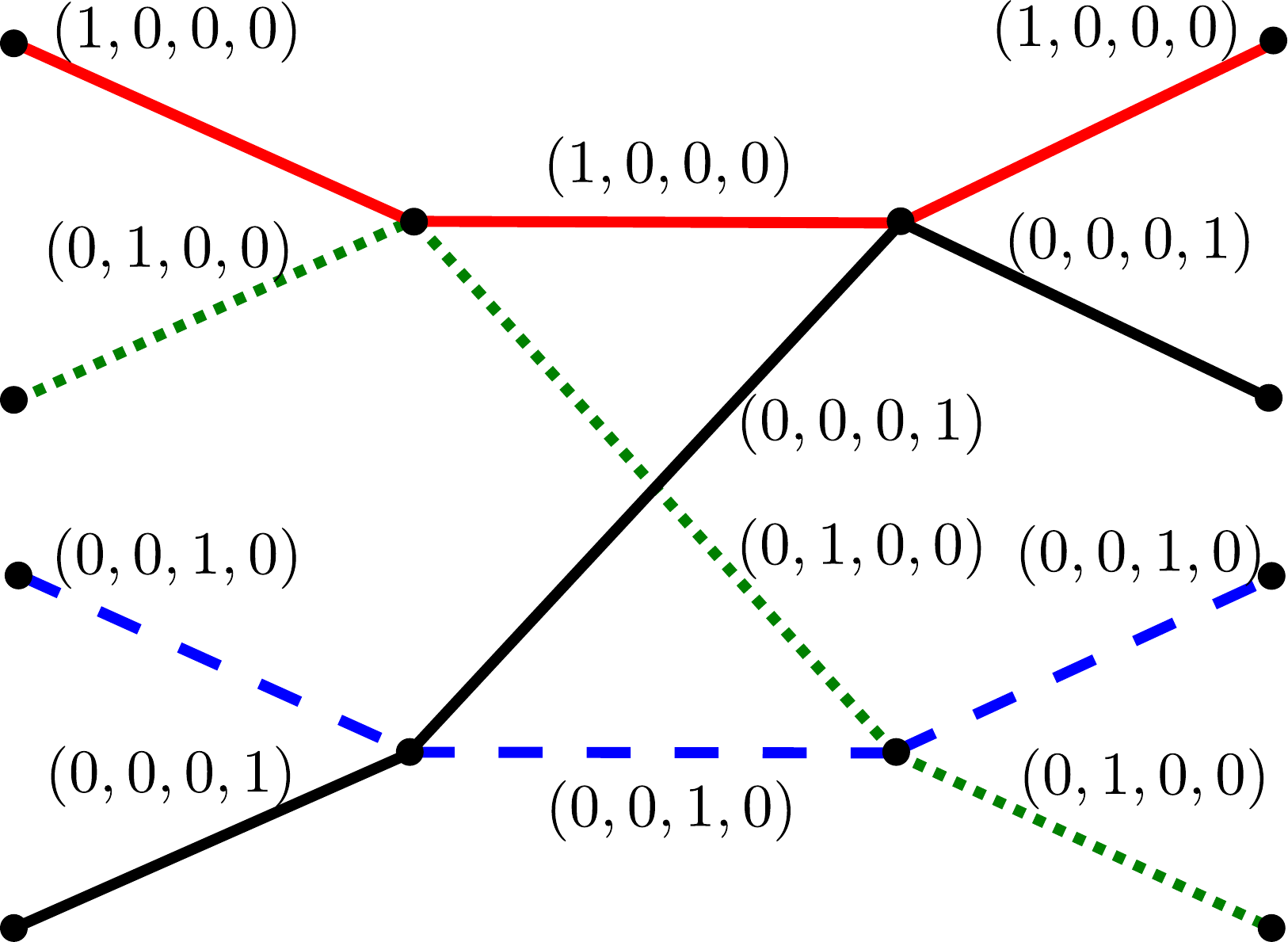}}
\hfill
        \subfigure[]{
    \includegraphics[height=2.8cm]{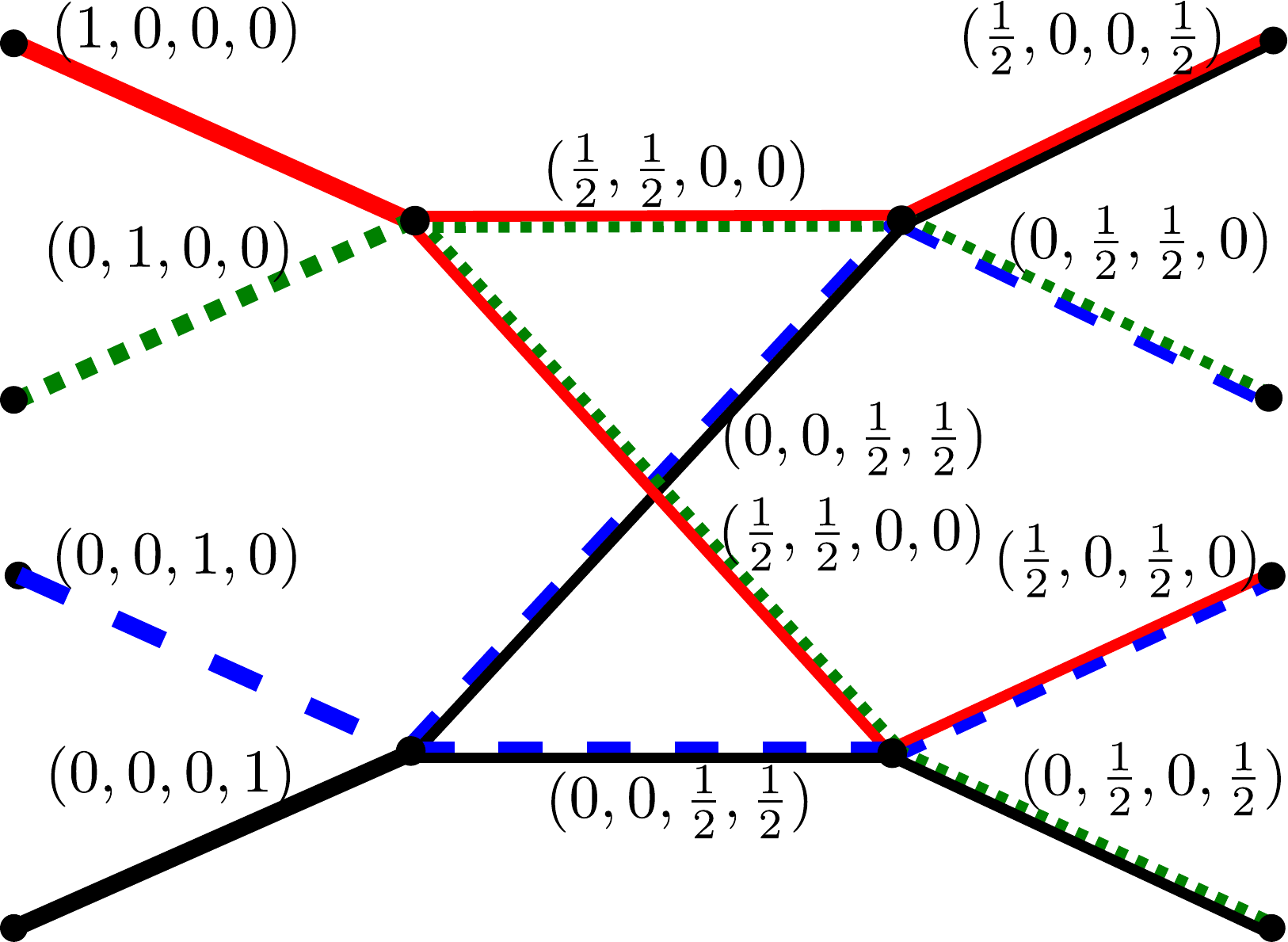}}
    \caption{Graph in which a fractional solution has cost greater
    than the cost of any integer solution.
    (a)~The input graph with 4 paths. The vector on the edge $e$
    corresponds to the weights $(w_e^1, w_e^2, w_e^3, w_e^4)$ for the
    paths on $e$.
    (b)~The optimum integer solution with $cost=15$.
    The vector on the edge $e$ corresponds to the solution
    $(x_e^1,x_e^2, x_e^3, x_e^4)$ on $e$.
    (c)~A fractional solution with $cost=16$.
    The vector on the edge $e$ corresponds to the solution
    $(x_e^1,x_e^2, x_e^3, x_e^4)$ on $e$.}
    \label{fig:noninteger}
\end{figure}

{\bf Linear Programming Formulation:}
Let $P_i$ denote the path in $G$ from $s_i$ to $t_j$ for some $j$ (the
path of the $i$-th ant).
For each $P_i$ and each edge $e$, we introduce a binary variable
$x_e^i$, which indicates whether
path $P_i$ passes through the edge $e$. The \SCoP{} problem can be
formulated as the following ILP:

$$
\begin{array}{llcllr}
\mbox{\textbf{maximize}}\:\: & \sum_e \sum_i w_e^i x_e^i \\
\mbox{\textbf{subject to }} & \sum_i x_e^i        &=&  1  & \forall
e\in E & \:\:\:\:(1)\\
& \sum_{uv} x_{uv}^i  &=&  \sum_{vw} x_{vw}^i \:\:& \forall v \in V \setminus \{s_1,t_1,\ldots, s_k,t_k\},
1\le i \le k & (2)\\
& \sum_{v} x_{s_iv}^i &=&  1 & \forall 1\le i \le k & (3) \\
& x_e^i & \in & \{0,1\} & \forall e\in E, 1\le i \le k & (4)
\end{array}
$$

Here constraint~(1) guarantees that the paths are
disjoint: there is exactly one path passing through every edge. Constraint~(2)
enforces consistency of the paths at
every intermediate vertex: if a vertex $v$ is contained in a path,
then the path passes through
an edge $(u,v)$ and an edge $(v,w)$ for some $u,w \in V$. In constraint~(3)
we sum over vertices $v$ with $(s_i,v)\in E$; it
implies that the $i$-th path
starts at source $s_i$.
If we relax the integrality constraint (4) by $0 \le x_e^i \le 1$,
we have a fractional LP formulation for the \SCoP{}
problem, which can be solved in polynomial-time. However,
the solution does not have a natural interpretation in the context of
ants (fractional ants do not make sense in the biological problem).
Further, we found an example for which the best
fractional solution has cost strictly greater than the cost of any
integer solution; see Fig.~\ref{fig:noninteger}.

We can convert an optimal fractional LP solution $x^*$ into a feasible integer solution as follows. Randomly
pick an ant $1 \le i \le k$, with probability of choosing the $i$-th
ant proportional to its weight $\sum_e w_e^i {x^*}_e^i$ in the fractional
solution. We then consider the graph with modified edge weights in which
the weight of an edge $e\in E$ is $w_e^i {x^*}_e^i$. We look for the longest path starting at source $s_i$ in this graph, assign the path to the $i$-th ant,
and remove this path from the graph. We then rerun the LP to
find a fractional solution on the smaller instance. Our experiments suggest that
this rounding scheme yields an integer solution that is close to optimal.

\section{Experimental Results}
We use a machine with Intel~i5 3.2GHz, 8GB~RAM and CPLEX Optimizer.

\begin{table}[t]
\centering
\caption{Average and worst root-mean-square error (in pixels) computed
for proposed algorithms.}
\label{tab:table2}
\medskip

\begin{tabular}{|l|c|c|c|}
\hline
Algorithm & Worst RMSE & Average RMSE & Runtime \\
\hline
Automated Solution      & 95.308    & 9.593 & $160$ min\\
Local Mean              & 105.271   & 12.531 & $<100$ ms \\
Local Median            & 112.741   & 9.801 & $<100$ ms \\
Local \Frechet{}        & 127.104    & 15.562 & $1.2$ sec \\
Homotopy Median         & 146.267   & 20.244 & $8.2$ sec \\
Buffer Median           & 171.556   & 23.998 & $9.7$ sec \\
Global ILP              & 20.588    & 8.716 & $34$ sec \\
Global Greedy           & 24.820    & 8.900 & $0.2$ sec \\
\hline
\end{tabular}
\end{table}

{\bf Real-World Dataset:}
Here we consider a real-world scenario and compare ground truth data
to seven different consensus algorithms described in this paper, along
with an automated solution.
To evaluate our various algorithms, we work with a video of a {\em
Temnothorax rugatulus} ant colony containing $10,000$ frames, recorded
at $30$ frames per second.
This particular video contains ants that are individually painted and
has been analyzed with the state-of-the-art automated multi-target tracking system of
Poff~{\em et~al.}~\cite{Poff12}. The method required about $160$ minutes to analyze
the video.
To evaluate the automated system,
they create a {\em ground truth} trajectory for each ant, by manually
examining \emph{every ant in every 100th frame} of the automated
output and correcting when necessary.
We use this ground truth data to evaluate the efficiency of the algorithms considered
and the automated system. Note that just by the
way the ground truth is generated, results are
inherently biased in favor of the automated solution.

Our dataset consists of $252$ citizen scientist generated
trajectories for $50$ ants, with between $2$ and $8$ trajectories per ant.
To compute the ant trajectories, we construct the interaction graph
$G$, as described in Section~\ref{sect:global}. For our dataset, the
graph contains $4,246$ vertices and $10,494$ edges.
We apply the five local methods and two of the global methods
(greedy and integer linear programming) to build consensus
trajectories.

We computed seven different consensus trajectories for each ant: five based on the local
algorithms, and two based on the global algorithms. An overview of
results is in Table~\ref{tab:table2} and in Fig.~\ref{fig:ex2}. We
compare all seven, as well as the trajectories computed by the
automated system, by measuring average and worst root-mean-square
error of the Euclidean distance between pairs of points of
computed  and ground truth trajectories. We notice that the approximate dimensions of an
ant in our video are $60\times 15$ pixels.

Among the local approaches, the local
median is best. The local mean is negatively impacted by the outliers in the data.
The \Frechet{} approach, Homotopy median, and Buffer median perform poorly.
This could be due to the very self-intersecting trajectories making these algorithms
miss entire pieces. We used the default values for all the parameters in
these algorithms; a careful tuning will likely improve accuracy.
In the \Frechet{} approach we tried $50$  different random orders of merging trajectories.

\begin{figure}[tb]
    \center
	\subfigure[]{
    \includegraphics[width=0.7\textwidth]{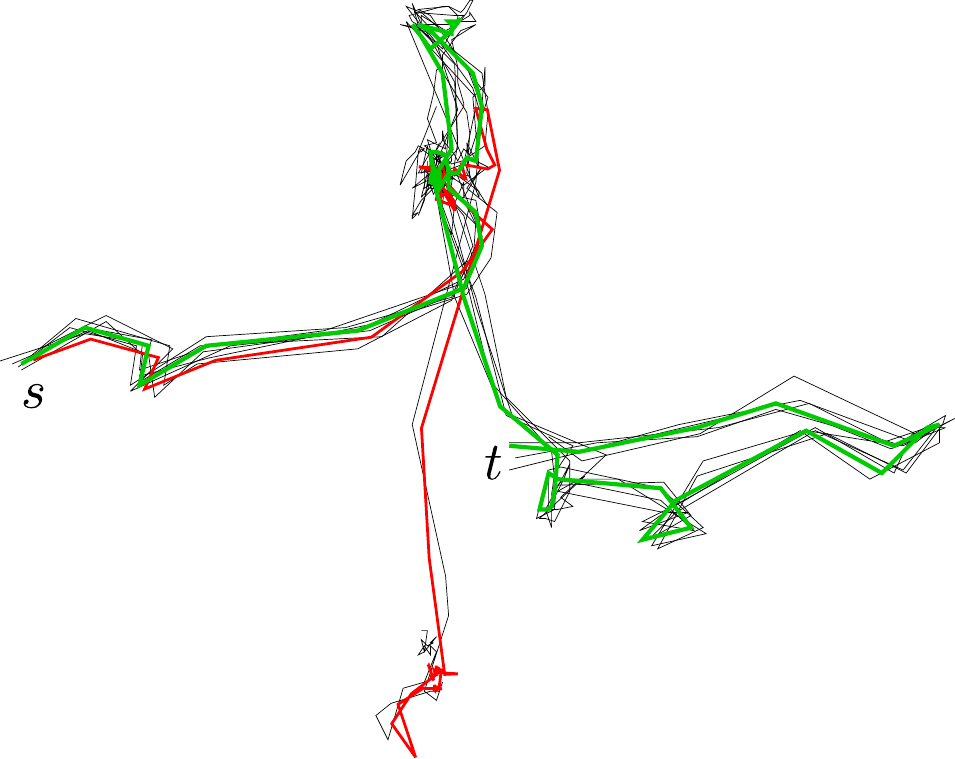}
   \label{fig:ex1}}
	\subfigure[]{
    \includegraphics[width=0.7\textwidth]{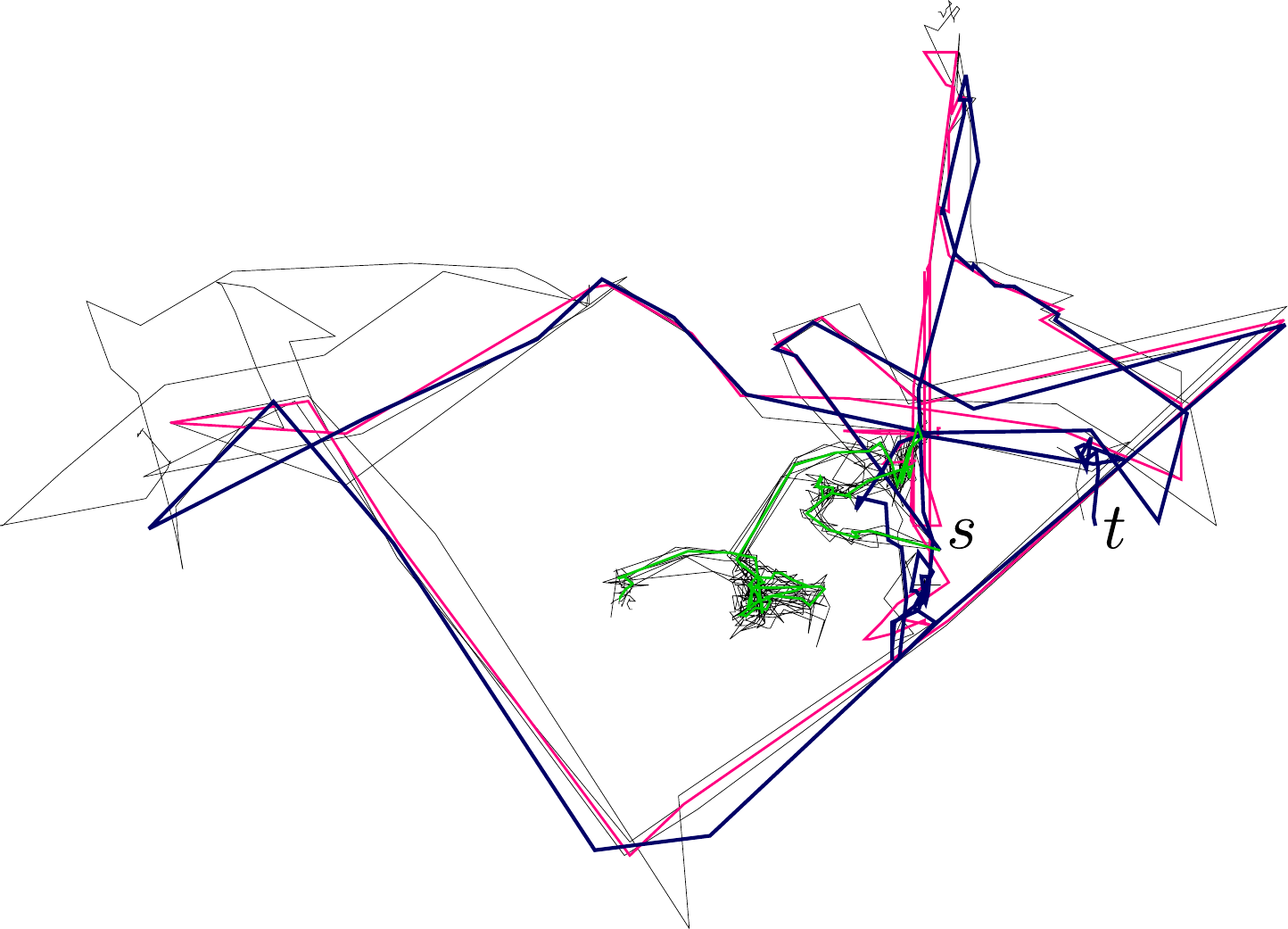}
    \label{fig:ex3}}
    \caption{\label{fig:exall}\small
(a)~The automated solution (red) switches to the wrong ant, while our
local median consensus (green) agrees the majority of input
trajectories (black).
    (b)~Most of the input trajectories (black) follow the wrong ant;
hence, the local median consensus (green) is incorrect. Here our
global consensus trajectory (dark blue) finds a trajectory closer to
the ground truth (pink).
}
\end{figure}

%

\begin{figure}[t]
    \center
	\subfigure[]{
    \includegraphics[width=0.48\textwidth]{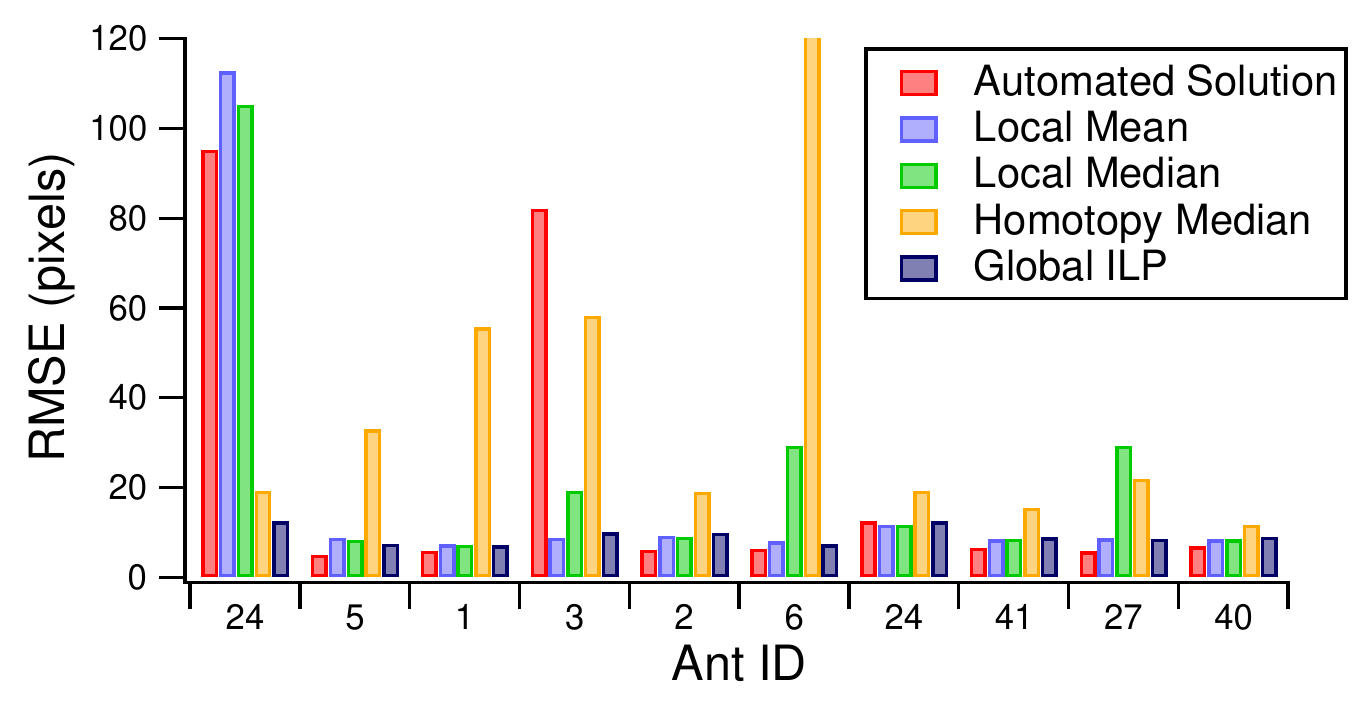}
    \label{fig:ex2}}
	\subfigure[]{
    \includegraphics[width=0.48\textwidth]{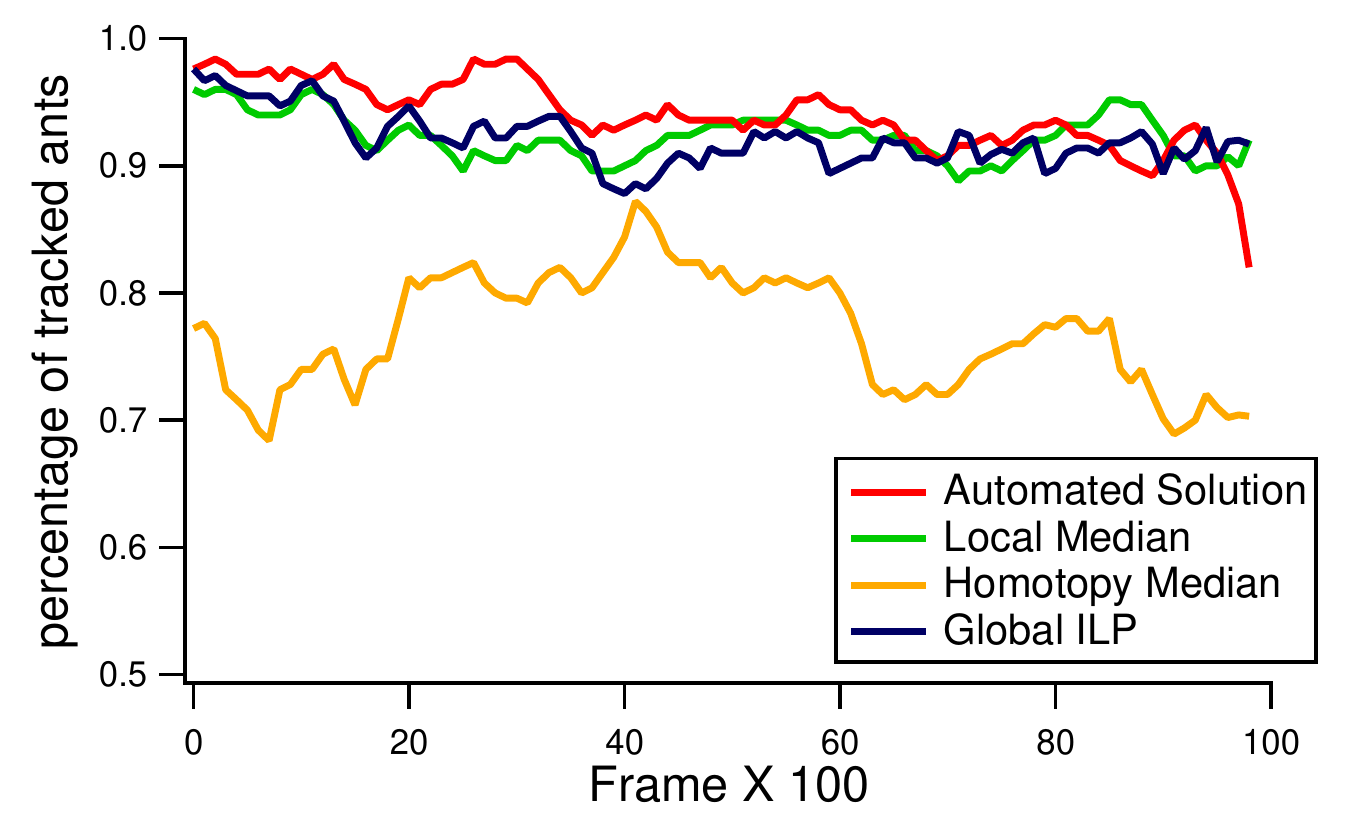}
    \label{fig:ex4}}
    \caption{
    (a)~Root-mean-square error for the 10 ``most movable'' ants
(according to the ground truth).
    (b)~Comparison of tracking accuracy between automated solution and
proposed algorithms.}
\end{figure}

The two global approaches perform similarly. There are only few
segments of trajectories where the results differ.
It is important to emphasize here the cases where the citizen
scientists make the same mistake do happen in practice.
We found an example in which 5 out of 8 input trajectories follow the
wrong ant; see Fig.~\ref{fig:ex3}. In this case,
none of the local algorithms have a chance to recover a correct
trajectory. Only the global approaches allow us to identify the
correct ant and produce the most accurate results.
We stress again that the ground truth data is inherently biased
towards the automated solution because it was obtained by modifying
the trajectories obtained from the automated solution. Yet our
global algorithms perform better.

A big challenge for the existing automated systems is tracking
ants in long videos. For long videos (e.g., hundreds or even thousands
of hours), automated tracking methods are not reliable. Whenever such
algorithms loose tracking, the error quickly accumulates and a
trajectory often is not recovered; see Fig.~\ref{fig:ex1}. Our global
approaches naturally resolve this problem. We evaluate
tracking accuracy as the percentage of ants correctly ``tracked'' at a
given timestamp; here, we consider the ant correctly tracked if the
distance between the ground truth and our trajectory is less than 15
pixels (typical width of ant head).
The accuracy of the automated solution decreases over time, and by the
end of the 5-minute video it is below $87\%$ accuracy; see
Fig.~\ref{fig:ex4}. Our algorithms are steadily over $90\%$ accuracy
over the entire video.

{\bf Synthetic Dataset:}
In order to validate our global approach on a
larger dataset, we generate a collection of synthetic graphs. Here
we test our algorithms for the \SCoP{} problem, that is,
the algorithms for computing disjoint paths on graphs, rather than for extracting optimal trajectories.
We construct a set of directed acyclic graphs having approximately the same
characteristics as the interaction graph computed for the real-world dataset.
The graph construction follows the same pipeline as described in Section~\ref{sect:global};
every graph is generated for $k\le 50$ ants and $T\le 100$ timestamps.
The ants form $k$ vertices for the initial timestamp; on every subsequent timestamp
pairs of ants meet with probability $0.4$ (the constant estimated for the real-world graph). Thus,
for every timestamp, we have from $k/2$ to $k$ vertices. The pairs of ants meeting at a timestamp
are chosen randomly with the restriction that indegrees and outdegrees of every vertex
in a graph are equal and at most $2$. By construction of the graphs, we naturally get
``ground truth'' paths for all ants.
Next we generate citizen scientist trajectories in two scenarios: $2-8$ trajectories (as in the real-world dataset) and $15-20$ trajectories per ant. A citizen scientist tracking an ant is modeled by a path starting at the source of the ant. At every junction vertex (with outdegree $2$) there is fixed probability $P(error)$ of making a mistake by switching from tracking the current ant to the other one.
If $P(error)=0$, then user trajectories always coincide with the ground truth paths; if $P(error)=0.5$, the trajectories may be considered as random walks on the graph.

\begin{figure}[t]
    \center
	\subfigure[]{
    \includegraphics[width=0.48\textwidth]{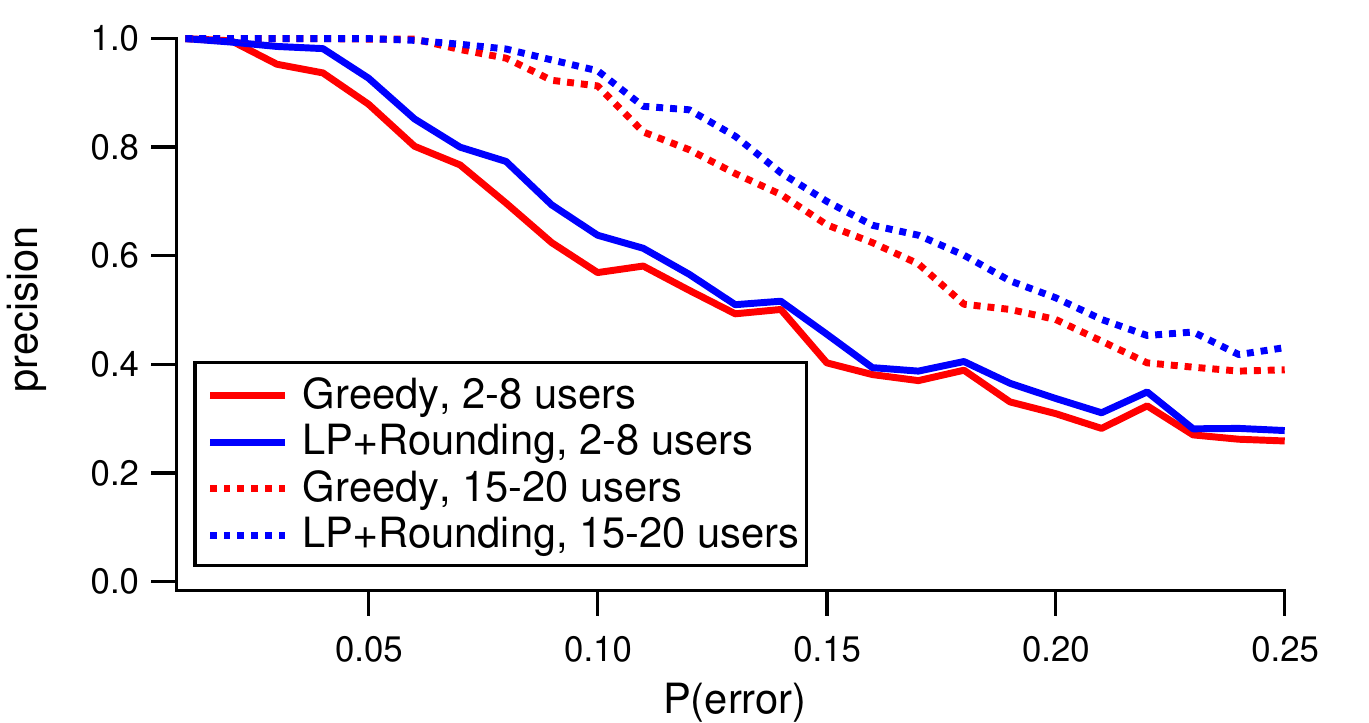}
    \label{fig:growing_real_a}}
\hfill
	\subfigure[]{
    \includegraphics[width=0.48\textwidth]{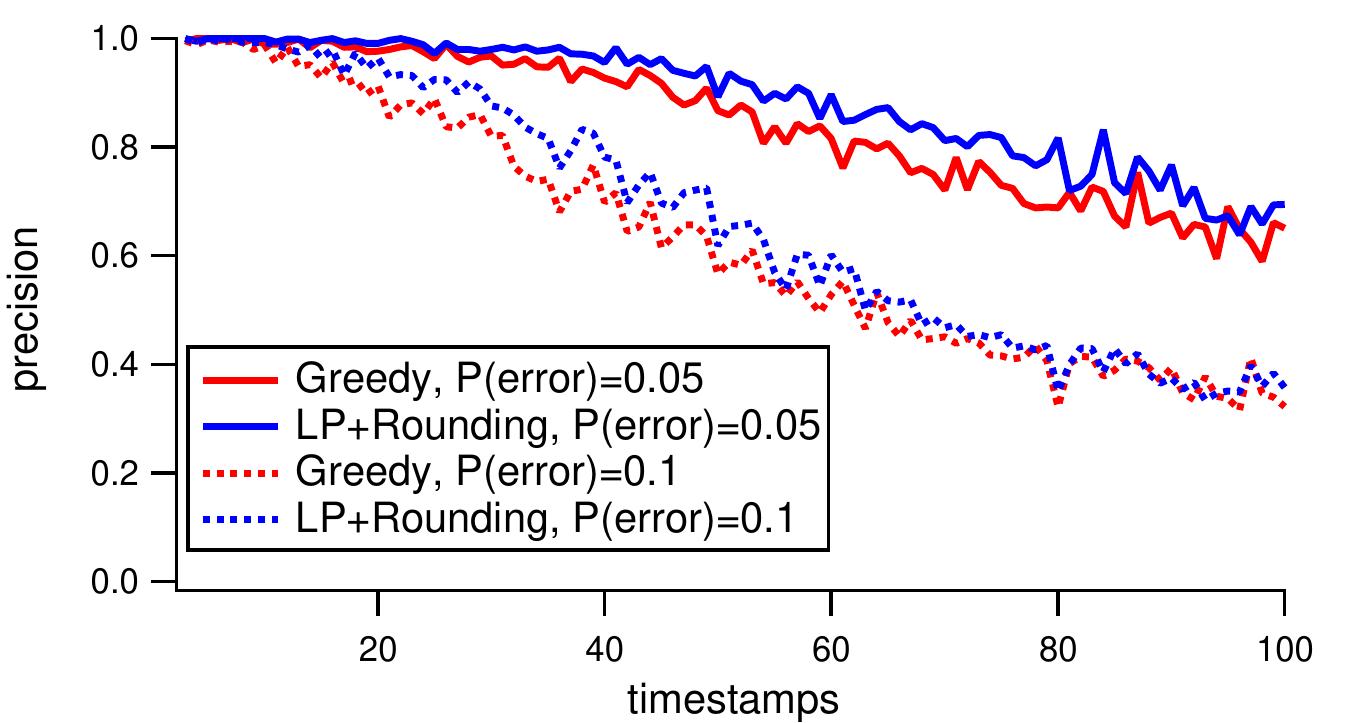}
    \label{fig:growing_real_b}}
    \caption{Precision of the algorithms on the synthetic dataset. Solid lines represent average values
    over $5$ runs of the algorithms for a given error/timestamp.
    (a)~Results for $k=50$ ants, $T=100$ timestamps (5-minute video segment), and
    various values of $P(error)$.
    (b)~Results for $k=50$ ants, $2-8$ user trajectories per ant, and various number of timestamps.}
    \label{fig:growing_real}
\end{figure}

We evaluate the greedy heuristic (\alg{Greedy}) and
the linear program with rounding (\alg{LP+Rounding}). As the latter
heuristic is a randomized algorithm, we report the best result over $5$ runs
for a given input. For small instances, we also compute
an exact solution using the integer linear program (\alg{ILP}).
We analyze the precision of the algorithms under various parameters. For every edge of the graph $G$,
we say that it is correctly identified if both the algorithm and the ground truth assign the edge to the same path.
The precision is measured as the fraction of correctly identified edges in $G$: a value of $1$ means
that all paths are correct.
As in the real-world dataset, we consider a scenario with $k=50$ ants. As expected, increasing the probability of making a mistake decreases the quality of the solution; see Fig.~\ref{fig:growing_real_a}. However,
increasing the number of user trajectories does help. Both algorithms recover all paths correctly
if the number of trajectories for each ant is more than $15$, even in a case with $P(error)=0.1$. On the other hand,
with $P(error)>0.2$ the precision drops to under $50\%$.
Although we cannot definitively measure the accuracy of all citizen scientists, empirical evidence from  our experiment indicates that the error rate was very low: $P(error) = 0.02$. This is an order of magnitude lower than the upper limit on errors that our algorithms can handle.
We also consider the impact of video length on the precision of the algorithms; see Fig.~\ref{fig:growing_real_b}.
Not surprisingly, precision is higher for short videos: for $2$-minute segments ($40$ timestamps) and $P(error)=0.05$,
\alg{LP+Rounding} produces the correct paths, while for $5$-minute segments ($100$ timestamps) only $60\%$ of paths are correct.

We also analyze the effectiveness of the \alg{Greedy} and \alg{LP+Rounding} algorithms for the \SCoP{} problem
with $P(error)=0.2$ and $2-8$ trajectories per ant;
see Fig.~\ref{fig:generated}.
To normalize the results, values are given as a percentage of the $cost$ of an optimal fractional solution (\alg{FLP}) for the \SCoP{} problem. Note that the \alg{ILP} results are in the range$[0.98,1.0]$, which means that an optimal integer solution is always very close to the optimal fractional solution.
Both \alg{Greedy} and \alg{LP+Rounding} perform very well, achieving $\approx 0.85$ of the optimal solution. These two
algorithms produce similar results, with \alg{LP+Rounding} usually outperforming \alg{Greedy}.

\begin{figure}[tb]
    \center
	\subfigure[]{
    \includegraphics[width=0.48\textwidth]{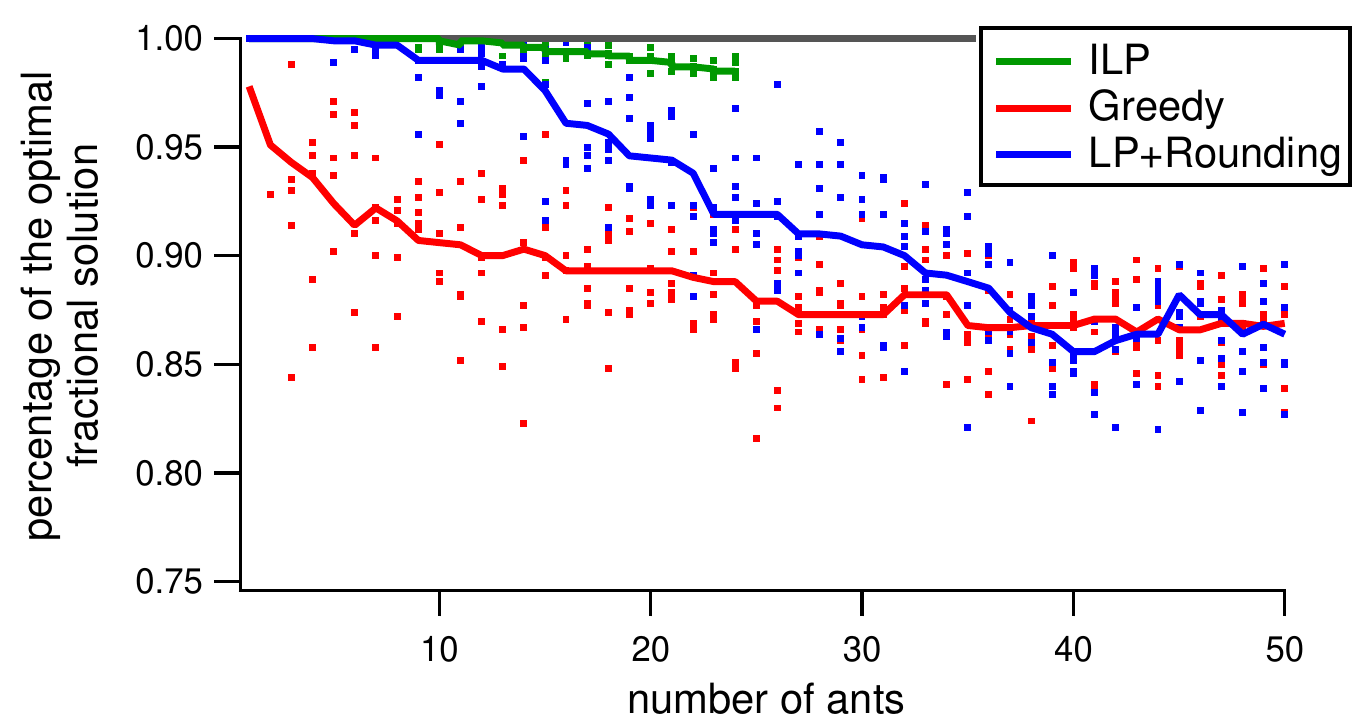}}
\hfill
	\subfigure[]{
    \includegraphics[width=0.48\textwidth]{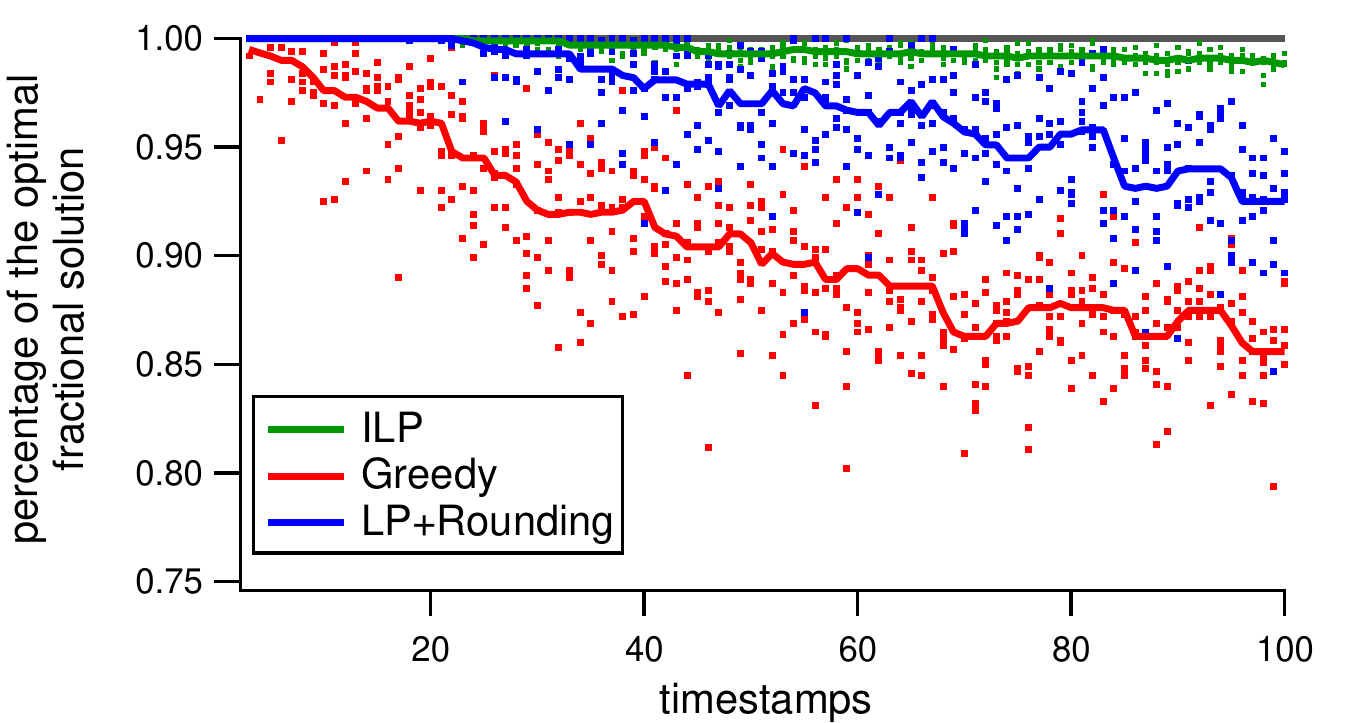}}
    \caption{Quality of the algorithms:
    ratio between the $cost$ of obtained solution and the optimal (fractional) $cost$.
    Results for single instances are depicted by dots and solid
    curves show the average values over $5$ runs
    for a given number of ants/timestamps.
    Note that the y-axis starts at value $0.75$.
    (a)~Results for $T=100$ timestamps, $P(error)=0.2$, and $2-8$ user trajectories per ant.
    (b)~Results for $k=20$ ants, $P(error)=0.2$, and $2-8$ user trajectories per ant.}
    \label{fig:generated}
\end{figure}

Running times are shown in Fig.~\ref{fig:generated-runtime}.
As expected, the greedy algorithm, with complexity dependent
linearly on the size of the graph, finishes in under few milliseconds.
The ILP is also relatively quick on the real-world dataset, computing the optimal solution within a minute.
On synthetic data and more erroneous real-world data
the
\alg{ILP} approach is applicable only when the number of ants is small, e.g., $k<25$. For larger values of $k$, the computation of optimal disjoint paths takes hours. On the other hand, \alg{LP+Rounding} is fast: the real-world
instances with $k=50$ ants and $T=100$ timestamps are processed within a minute. \alg{Greedy} takes only $2-3$ seconds on the largest instances and may be used in an online fashion. We conclude that the running times of all of our algorithms (except \alg{ILP}) are practical.

\section{Conclusions and Future Work}
\label{sect:open}
We described a system for computing consensus trajectories
from a large number of  input trajectories,
contributed by untrained citizen scientists.
We proposed a new global approach for computing consensus
trajectories and
experimentally demonstrated its effectiveness.
In particular, the global approach outperforms the state-of-the-art in
computer vision tools, even in their most advantageous setting (high
resolution video, sparse ant colony, individually painted ants).
In reality, there are hundreds of thousands of hours of video in settings
that are much more difficult for the computer vision tools and where
we expect our citizen science approach to compare even more
favorably.

A great deal of challenging problems remain. Arguably, the best method
would be to track ``easy ants'' and/or ``easy trajectory segments''
automatically, while asking citizen scientists to solve the hard ants
and hard ant trajectory segments. In such a scenario, every input
trajectory will describe a part of the complete ant trajectory, which
would require stitching together many short
pieces of overlapping trajectories.

\begin{figure}[t]
    \center
	\subfigure[]{
    \includegraphics[width=0.48\textwidth]{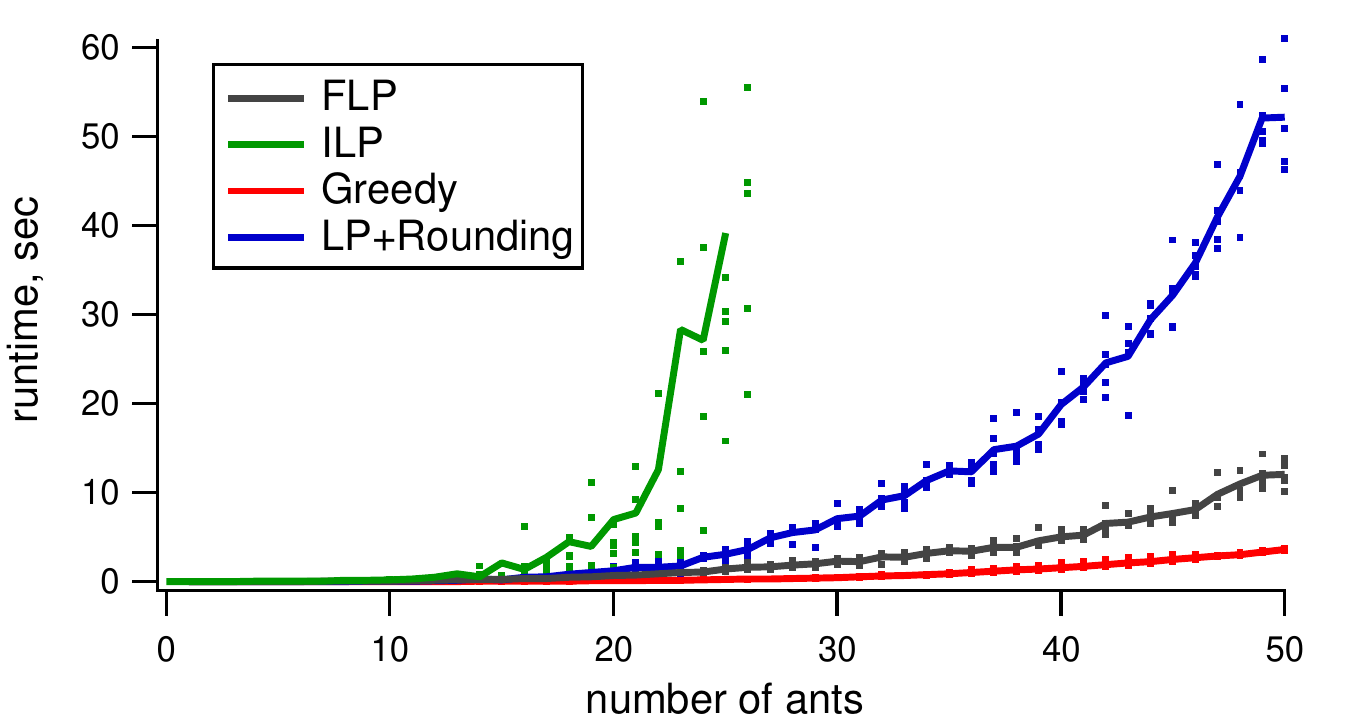}}
\hfill
	\subfigure[]{
    \includegraphics[width=0.48\textwidth]{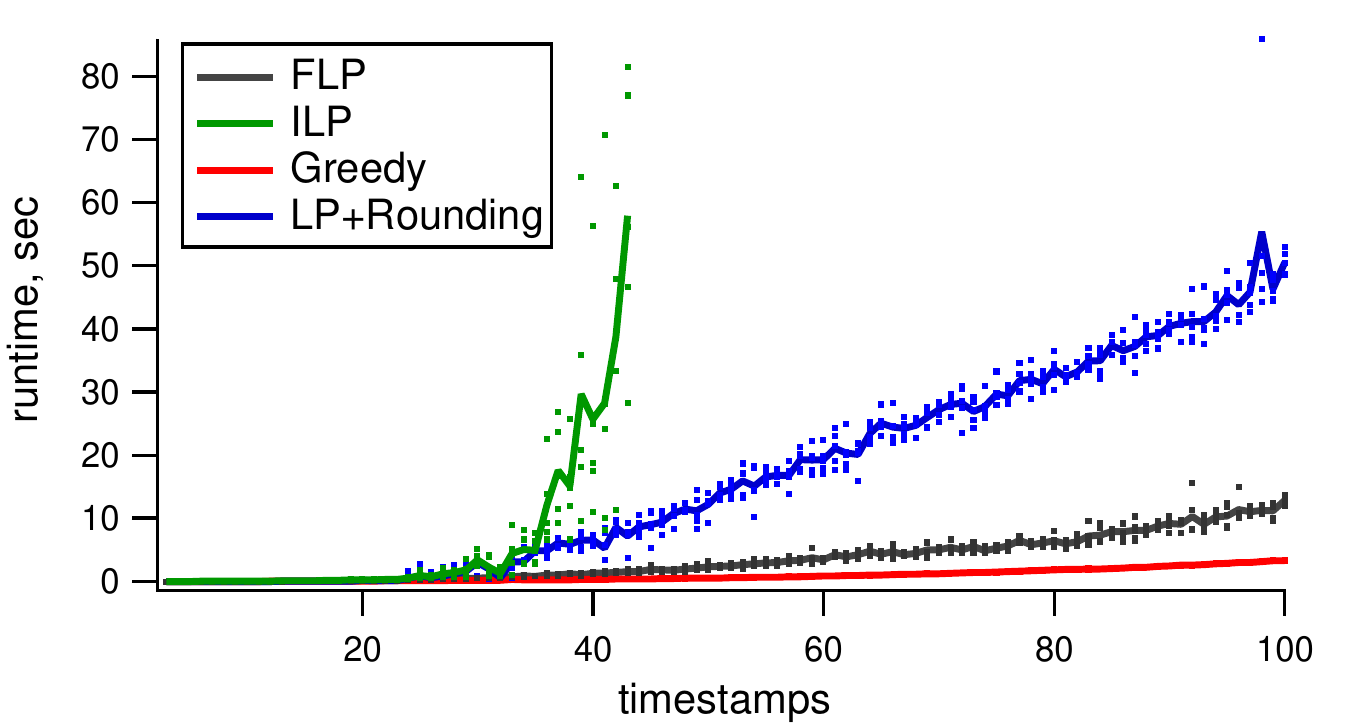}}
    \caption{Running time of the algorithms with $P(error)=0.2$ and $2-8$ trajectories per ant.
    The results for single instances are depicted by dots, while solid lines represent average values over $5$ runs
    for a given number of ants/timestamps.
    (a)~Results for $T=100$ timestamps and various number of ants.
    (b)~Results for $k=50$ ants and various length of a video.}
    \label{fig:generated-runtime}
\end{figure}

\medskip\noindent {\bf Acknowledgments.}
We thank the A.~Dornhaus lab for introducing us to the
problem, and M.~Shin and T.~Fasciano for automated solutions and ground truth. We thank A.~Das, A.~Efrat, F.~Brandenburg, K.~Buchin, M.~Buchin, J.~Gudmundsson, K.~Mehlhorn, C.~Scheideler, M.~van
Kreveld, and C.~Wenk for discussions.
Finally, we thank J.~Chen, R.~Compton, Y.~Huang, Z.~Shi, and
Y.~Xu for help with the game development.

\bibliographystyle{splncs03}
\bibliography{refs}

\end{document}